\documentclass{amsart}
\usepackage{amssymb, latexsym, cmll, stmaryrd,eqnarray}
\usepackage{mathtools}
\usepackage{xspace}
\usepackage{multirow}
\usepackage{amsaddr}
\usepackage{booktabs}   
\usepackage{subcaption} 

\begin{document}
%Preambuła:
%\usepackage{unicode-math}
%\usepackage{mathspec}
%\usepackage{todonotes}
%\newcommand{\TODO}[1]{{\todo[inline]{#1}}}

\newcommand{\blue}[1]{{\color{blue}{#1}}}
%Notatka na marginesie zaczepiona w tekście:

%the length of the longest one beginning at $A[i]$ and $B[j]$. \todo{Do we use this convention anywhere?} If $A[i] \neq B[j]$, we say that $\lcisp_{A, B}(i, j) = \lciss_{A, B}(i, j) = 0$.

%Notatka jako osobny akapit:

%\TODO{Write proof of Lemma 6}

%%%%%%%%%%%%%%%%%%%% structure

\newcommand{\ssection}{\section}

\newcommand{\note}[1]{\marginpar[\hfill{\tiny #1}]{\tiny #1}}

\newtheorem{lm}{Lemma}[section]
\newtheorem{thm}[lm]{Theorem}
\newtheorem{prp}[lm]{Proposition}
\newtheorem{cor}[lm]{Corollary}
\newtheorem{rmk}[lm]{Remark}
\newtheorem{ex}[lm]{Example}
\newtheorem{df}[lm]{Definition}
\newtheorem{prob}{Problem}
\newtheorem{conj}[lm]{Conjecture}%[prob]
\newtheorem{fact}[lm]{Fact}

\newcommand{\myqed}{\hfill$\Box$}

\newcommand{\wyj}[1]{\hfill{\small #1}}
%\newenvironment{proof}{\noindent{\it Proof}. }{\myqed \\}

%%%%%%%%%%%%%%% Numbering
\newcounter{senumi}[section]
\newcounter{senumip}[section]
\newcounter{temp}[section]

\def\thesenumi{\thesection.\arabic{senumip}}
\def\p@senumip\thesenumip{\thesenumi}
\newenvironment{senumerate}%
    {\begin{list}%
        {(\thesenumi)}%
        {\usecounter{senumip}}
        %{\usecounter{senumip}\setlength{\rightmargin}{\leftmargin}}
        \setcounter{senumip}{\value{temp}}
    }%
    {\setcounter{temp}{\value{senumip}}
     \end{list}}
%%%%%%%%%%%%%%%%%%%%%%%%%%%%%%%%%%%%%%%%%%%%%%%%%%%
\newcounter{penumi}[section]
\newcounter{penumip}[section]
\newcounter{ptemp}[section]

\def\thepenumi{\arabic{penumi}}
\newenvironment{penumerate}%
   {\begin{enumerate}%
        \setcounter{penumi}{\value{ptemp}}%
        \setcounter{enumi}{\value{ptemp}}%
    }%
    {\setcounter{ptemp}{\value{enumi}}
     \end{enumerate}}
%%%%%%%%%%%%%%%%%%%%%%%%%%%%%%%%%%%%%%%%%%%%%%%%%%%%%%
\newcounter{ppenumi}[section]
\newcounter{ppenumip}[section]
\newcounter{pptemp}[section]

\def\theppenumi{\theptemp.\arabic{ppenumi}}
\newenvironment{ppenumerate}%
    {\begin{list}%
        {(\theppenumi)}%
        {\usecounter{ppenumi}\setlength{\rightmargin}{\leftmargin}}
        \setcounter{ppenumi}{\value{pptemp}}
    }%
    {\setcounter{pptemp}{\value{ppenumi}}
     \end{list}}

\newcounter{entmp}
%%%%%%%%%%%%%%%%%%%%%%%%%%%%%%%%%%%%%%%%%%%%%%%%%%%%%%%%%%%%%%%%%%%%%%%%%%%%
%%%%tenumerate counter to create lists with (1),(2), etc.

\newenvironment{tenumerate}{\begin{enumerate}}{\end{enumerate}}
%%%%%%%%%%%%%%%%%%%%%%%%%%%%%%%%%%%%%%%%%%%%%%%%%%%%%%%%%%%%%%%%%%%%%%%%%%%%%%

%%%%%%%%%%%%%%%%%%%%%%%%%%%%%%
%% Notational conventions:
%% A is an algebra
%% a,b typical members of A
%% T is an index set
%% s,t typical members of T
%% D is a diagonal subalgebra of A^T
%% c,d and f,g are typical members of D, sometimes u,v,w used.
%% G is a generating set
%% X is the generating set of a free algebra
%% \mu is the monolith of SI
%% \alpha,\beta,\gamma etc. are congruences.
%% \alpha is usually an atom
%% \Cal V is a variety
%% bold e used for idempotent.

%%%%%%%%%%%%%%%%%%%%%%%%%%%%%%
%%\m produces models, eg. \m a gives algebra A
\newcommand{\m}[1]{{\uppercase {\mathbf{#1}}}}
%%\rel produces relational strucrtures, eg. \rel a  A
\newcommand{\rel}[1]{{\uppercase {\mathbb{#1}}}}
\newcommand{\mrel}[2]{{\uppercase{\mathbf{#1}}}\!\left[\uppercase{\mathbb{#2}}\right]}

%%\polsat
%\newcommand{\ceqv}[1]{\ensuremath{\operatorname{\textsc{Ceqv}%\textsc{C-EQV}
%                                \ifthenelse{\equal{#1}{}}{}{\!\left( {\m #1} \right)}}}}
\newcommand{\ceqv}[1]{\ensuremath{\operatorname{\textsc{\upshape{Ceqv}}}%\textsc{C-EQV}
                                \ifthenelse{\equal{#1}{}}{}{\!\left( { #1} \right)}}}
\newcommand{\csat}[1]{\ensuremath{\operatorname{\textsc{\upshape{Csat}}}%\textsc{C-SAT}
                                \ifthenelse{\equal{#1}{}}{}{\!\left( { #1} \right)}}}
\newcommand{\scsat}[1]{\ensuremath{\operatorname{\textsc{\upshape{SCsat}}}%\textsc{C-SAT}
                                \ifthenelse{\equal{#1}{}}{}{\!\left( { #1} \right)}}}
\newcommand{\Csat}[1]{\ensuremath{\operatorname{\textsc{\upshape{Csat}}}%\textsc{C-SAT}
                                \ifthenelse{\equal{#1}{}}{}{\!\left( {\m #1} \right)}}}
\newcommand{\Ceqv}[1]{\ensuremath{\operatorname{\textsc{\upshape{Ceqv}}}%\textsc{C-SAT}
                                \ifthenelse{\equal{#1}{}}{}{\!\left( {\m #1} \right)}}}
\newcommand{\mcsat}[1]{\ensuremath{\operatorname{\textsc{\upshape{MCsat}}}%\textsc{Multi-C-SAT}
                                \ifthenelse{\equal{#1}{}}{}{\!\left( {\m #1} \right)}}}
\newcommand{\SCsat}[1]{\ensuremath{\operatorname{\textsc{\upshape{SCsat}}}%\textsc{C-SAT}
                                \ifthenelse{\equal{#1}{}}{}{\!\left( {\m #1} \right)}}}

\newcommand{\csp}[1]{\ensuremath{\operatorname{\textsc{\upshape{CSP}}}
                                \ifthenelse{\equal{#1}{}}{}{\!\left( {#1} \right)}}}

\newcommand{\polsatstar}{\csat}%[1]{\textsc{PolSat$\star$}\left( {\m #1} \right)}
\newcommand{\cpolsatstar}{\csat}%[1]{\textsc{PolSatC$\star$}\left( {\m #1} \right)}
%%%%%% end of polsat stuff

%%NP-complete
\newcommand{\npc}{\textsf{NP}-complete\xspace}
\newcommand{\conpc}{\textsf{co-NP}-complete\xspace}
\newcommand{\np}{\textsf{NP}\xspace}
\newcommand{\ptime}{\textsf{P}\xspace}
\newcommand{\pspace}{\textsf{PSPACE}\xspace}
\newcommand{\usp}{Uniform Solution Property\xspace}%property of solutions in DL
\newcommand{\USP}{USP\xspace}

%%\vr produces varieties, eg. \vr v gives variety V
\newcommand{\vr}[1]{{\uppercase {\mathcal {#1}}}}

\newcommand{\set}[1]{{\left\{ {#1} \right\} }}
\newcommand{\ci}{\subseteq}
\newcommand{\co}{\supseteq}

\newcommand{\card}[1]{\left| #1 \right|}
\newcommand{\cardd}[1]{\# #1}
\newcommand{\equa}[1]{\left\| #1 \right\|}
\newcommand{\tup}[1]{\langle #1 \rangle}

\newcommand{\intv}[2]{I\left[#1,#2\right]}
\newcommand{\vpair}[2]{{{#1}\choose{#2}}}

\renewcommand{\leq}{\leqslant}
\renewcommand{\geq}{\geqslant}
\renewcommand{\le}[1]{\leqslant_{#1}}
\renewcommand{\ge}[1]{\geqslant_{#1}}

\newcommand{\comp}{\leq\!\geq}
\newcommand{\scomp}{<\!>}

\newcommand{\dist}[2]{{\sf dist}\!\left( #1,#2 \right)}
\newcommand{\distt}[3]{{\sf dist}_{#3}\!\left( #1,#2 \right)}

\renewcommand{\mapsto}{\longmapsto}
\newcommand{\tomaps}{\longmapsfrom}

\newcommand{\join}{\vee}
\newcommand{\meet}{\wedge}
\newcommand{\jjoin}{\bigvee}
\newcommand{\mmeet}{\bigwedge}
%bar convention commands
\newcommand {\bc}[1]{{\overline {#1}} }

%%\con gives congruence operator
%%\Cn gives congruence lattice in boldface
%%\cn gives congruence lattice in regular
%%\cg gives cong generated in alg. by two elements
%%\cgt gives cong generated by a set of ordered pairs
\newcommand{\con}[1]{{\sf Con\:\m{#1}}}
\newcommand{\cn}[1]{{\sf Con\:\m{#1}}}
\newcommand{\Cn}[1]{{{\sf Con\:}\m{#1}}}
\newcommand{\cg}[3]{{\rm Cg}^{{\m {#1}}}({#2},{#3})}
\newcommand{\cgt}[3]{{\rm Cg}^{{\m {#1}}}({#2} \times {#3})}

\newcommand{\Pw}[1]{{\mathcal P} ({#1})}
\newcommand{\Pwp}[1]{{\mathcal P}^+ ({#1})}

\newcommand{\pol}[1]{{\rm Pol\:\m #1}}
\newcommand{\Pol}[1]{{\rm Pol( #1)}}
\newcommand{\poln}[2]{{\rm Pol}_{#1}\m #2}
\newcommand{\clo}[1]{{\rm Clo\:\m #1}}

\renewcommand{\d}{\po D}
\newcommand{\q}{\po Q}

%commands and functions for tame congruence theory
% \id for idempotent, \tn for type number
% \typ for roman typ  \rst for restriction (2 arguments)
\newcommand{\po}[1]{{\mathbf{#1}}}
\newcommand{\te}[1]{{\mathbf{#1}}}
\newcommand{\tn} [1]{{\mathbf{#1}}}
\newcommand{\typ}{{\rm typ}}
\newcommand{\typset}[1]{\typ\set{#1}}
\newcommand{\rst}[2]{ {#1} |_{{#2}} }
\newcommand{\charr}{{\rm char}}
\newcommand{\charrset}[1]{\charr\set{#1}}

\newcommand{\prect}[1]{\prec_{\tn #1}}

\newcommand{\minim}[3]{M_{\m #1}\left(#2,#3\right)}

\newcounter{note}
\newcounter{claim}

%This overlines one character without needing extra braces.
\renewcommand{\o}[1]{\overline {#1}}
\newcounter{ttable}
\newcommand{\ttable}{\refstepcounter{ttable}{Table \arabic{ttable}:\ \ }}

\newcommand{\centr}[2]{\left(#2 : #1\right)}
\newcommand{\comm}[2]{\left[ #1 , #2 \right]}
\newcommand{\com}[2]{\left[ #1 , #2 \right]}
\newcommand{\commm}[2]{\left[ #1 , \ldots, #2 \right]}
\newcommand{\commmk}[3]{\left[ #1 , \ldots, #2 \right]_{#3}}

\newcommand{\efdef}{{\rm ($\star$)}}

\newcommand{\map}{\longrightarrow}

\newcommand{\congruent}[1]{\stackrel{#1}{\equiv}}

\newcommand{\h}[1]{\widehat{#1}}
\newcommand{\oo}[1]{\overrightarrow{#1}}

\newcommand{\hsp}{{\sf {HSP}}}
\newcommand{\hs}{{\sf {HS}}}

%%%%% new??
\newcommand{\fj}{\varphi}
\newcommand{\epsi}{\varepsilon}

%%%%%%%%%%%%%%%%%%%%%%%%%%%%%%%%%%%%%%%%%%%%%%%%%%%%%%%%%%%%%%%%%%%%%%%%
%%%%%%%%%%%%%%%%%%%%%%%%%%%%%%%%%%%%%%%%%%%%%%%%%%%%%%%%%%%%%%%%%%%%%%%%
%%%%%%%%%%%%%%%%%%%%%%%%%%%%%%%%%%%%%%%%%%%%%%%%%%%%%%%%%%%%%%%%%%%%%%%%
\newcommand{\dpp}[1]{\m D[{#1}]}
\newcommand{\dhh}{\dpp{p_1,\ldots,p_h}}
\newcommand{\dpq}{\m D[p,q]}
\newcommand{\dd}{\m D}

\newcommand{\cs}{\card{\Gamma}} %rozmiar circuitu
\newcommand{\pand}{\textsf{AND}}
\newcommand{\por}{\textsf{OR}}
\newcommand{\pnot}{\textsf{NOT}}
\newcommand{\cmod}{\textsf{MOD}}
\newcommand{\bb}{\textsf{b}}
\newcommand{\pcnf}{\textsf{CNF}}
\newcommand{\pv}{\textsf{V}}
\newcommand{\emptys}{\emptyset}
\newcommand{\dstar}{{\star\star}}
\newcommand{\pstar}{{\dagger}}
\newtheorem*{conjj}{Conjecture}

\newcommand{\scom}[2]{{#1}_{[#2]}}

%%%%%%%%%%%%%%%%%%%%%%%%%%%%%%%%%%%%%%%%%%%%%%%%%%%%%%%%%%%%%%%%%%%%%%%%
%%%%%%%%%%%%%%%%%%%%%%%%%%%%%%%%%%%%%%%%%%%%%%%%%%%%%%%%%%%%%%%%%%%%%%%

\newcommand{\polsat}[1]{\ensuremath{\operatorname{\textsc{\upshape{PolSat}}}
                                 \ifthenelse{\equal{#1}{}}{}{\!\left( { #1} \right)}}}
\newcommand{\poleq}[1]{\ensuremath{\operatorname{\textsc{\upshape{PolEqv}}}
                                 \ifthenelse{\equal{#1}{}}{}{\!\left( { #1} \right)}}}

\title{Intermediate problems in modular circuits satisfiability}        

\author{Pawe\l{} Idziak, Piotr Kawa\l{}ek}\thanks{The project is partially supported  by Polish NCN Grant\break \# 2014/14/A/ST6/00138.}
\address{Jagiellonian University, Faculty of Mathematics and Computer Science, Department of Theoretical Computer Science ul.~Prof.~S.~\L{}ojasiewicza~6,~30-348,~Krak\'ow, Poland}
\email{idziak@tcs.uj.edu.pl,  piotr.kawalek@doctoral.uj.edu.pl}

\author{Jacek Krzaczkowski}
\address{Maria Curie-Sklodowska University, Faculty of Mathematics, Physics and Computer Science, Department of Computer Science ul.~Akademicka~9,~20-033,~Lublin, Poland}
\email{krzacz@poczta.umcs.lublin.pl}

\begin{abstract}
In \cite{ik:lics18}
a generalization of Boolean circuits to arbitrary finite algebras had been introduced and applied to sketch  \ptime versus \npc borderline for circuits satisfiability over algebras from congruence modular varieties.
However
%nilpotent but not supernilpotent algebras had not been put on any side of this borderline.
the problem for nilpotent (which had not been shown to be NP-hard)
but not supernilpotent algebras (which had been shown to be polynomial time) remained open.

In this paper we provide a broad class of examples, lying in this grey area,
and show that,
under the Exponential Time Hypothesis and Strong Exponential Size Hypothesis
(saying that  Boolean circuits need exponentially many modular counting gates to produce boolean conjunctions of any arity),
satisfiability over these algebras have intermediate complexity between
$\Omega(2^{c\log^{h-1} n})$ and $O(2^{c\log^h n})$,
where $h$ measures how much a nilpotent algebra fails to be supernilpotent.
We also sketch how these examples could be used as paradigms to fill the nilpotent versus supernilpotent gap in general.

Our examples are striking in view of the natural strong connections between circuits satisfiability and Constraint Satisfaction Problem for which the dichotomy had been shown
by Bulatov \cite{bulatov} and Zhuk \cite{zhuk}.
\end{abstract}

\keywords{circuit satisfiability, 
intermediate problems,
solving equations,
constraint satisfaction problem}  %

\maketitle

\section{Introduction \label{sect-intro}}

In \cite{ik:lics18} a generalization of Boolean circuits
to multi-valued ones had been introduced.
This concept was formalized by defining circuits over arbitrary finite algebra $\m A$.
Then the computational complexity of the following problems was considered:
\begin{itemize}
  \item $\csat{\m A}$ -- circuits satisfiability over the algebra $\m A$,
  \item $\scsat{\m A}$ -- simultaneous satisfiability of a set of circuits
                          over the algebra $\m A$,
  \item $\ceqv{\m A}$ -- circuits equivalence over the algebra $\m A$.
\end{itemize}
This has been done by treating the (basic) gates of the circuits
as fundamental operations of the corresponding algebra $\m A$,
while the universe $A$ of this algebra consists
of the possible values on inputs and output of the gates.
Such translation has been shown to preserve the complexity when passing respectively between
\begin{itemize}
  \item $\csat{\m A}$ and deciding if an equation over $\m A$ has a solution,
  \item $\scsat{\m A}$ and deciding if a system of equations over $\m A$ has a solution,
  \item $\ceqv{\m A}$ and deciding if two polynomials determine the same function over $\m A$,
\end{itemize}
but with the possibility of endowing algebra $\m A$ with finitely many additional operations that are already definable in $\m A$.
Such (finite) expansions allows to concentrate on the algebraic structure of the considered algebras in order to classify them with respect to computational complexity of the above problems.
Making the algebra independent of its basic operations is crucial,
as for example equation solving over the group $\m S_3$ can be done in \ptime,
while for the same group endowed with definable operation resembling binary commutator
$[x,y]=x^{-1}y^{-1}xy$ the very same problem became \npc.
Actually \cite{ik:lics18} presents an attempt to such classification
for a very broad class of algebras
covering most of the ones considered in mathematics and computer science,
like groups, rings, modules, lattices, Heyting algebras and many other algebras arising from logic.
The restriction put for those algebras was that they have to belong to congruence modular varieties.
This assumption made it possible to use advanced tools of universal algebras that work in such a setting.
Under this additional assumption it has been shown that
if an algebra $\m A$ fails to decompose nicely,
i.e. into a direct product of a nilpotent algebra and an algebra that essentially is a subreduct
of a distributive lattice then $\csat{}$ for $\m A$ (or at least one of its quotients) is \npc.
And, almost conversely, if $\m A$ does decompose nicely (in the above sense),
but with the additional assumption that the nilpotent factor is actually supernilpotent,
then $\csat{\m A}$ is \ptime.
Very similar statements hold for $\ceqv{\m A}$,
but the `lattice' factor disappears here
as $\ceqv{}$ for distributive lattices is \npc.

Although the problems $\csat{}$ or $\SCsat{}$ resemble Constraint Satisfaction Problem,
there are some subtle differences here, so that the CSP dichotomy shown by Bulatov \cite{bulatov} and Zhuk \cite{zhuk} can not be used for $\csat{}$.
As it was noticed in \cite{ik:lics18} the problems $\SCsat{}$ and $\csp{}$
can be bisimulated in polynomial time,
namely each finite algebra $\m A$ can be transformed into a finite relational structure $\rel D$,
and each finite relational structure $\rel D$ can be translated into a finite algebra $\m A$ so that
the problems $\SCsat{\m A}$ and $\csp{\rel D}$ are equivalent.
Surprisingly for a single circuit/equation such kind of translation works only when going from relational structures to algebras.
In fact, under some complexity hypothesis (like for example ETH)
this paper shows that the other way is blocked.
As we will see, this is due to the hardness of incorporating arbitrary long conjunction
(between constraints) by translating them into relatively short polynomials of an algebra.

As it is easily seen the nice pre-characterization of algebras with $\csat{}/\ceqv{}$ solvable in polynomial time leaves the nilpotent but not supernilpotent gap which is unsolved.
The essential difference between this two concepts of nilpotency
lies in fact that in supernilpotent algebras there is an absolute bound for the arity
of expressible (by polynomials) conjunction.
In nilpotent (but not supernilpotent) algebras conjunction-like polynomials
of arbitrary arity $n$ do always exist
but the known ones are too long to be used to polynomially code \npc problems in $\csat{}$.
In section \ref{sect-strat} we split nilpotent algebras into slices
that will correspond to the measure how much a nilpotent algebra fails to be supernilpotent.
This distance $h$ is determined by the behavior of a multi-ary commutator operation
on congruences of $\m A$ which is used to define $h$-step supernilpotent algebras.
On the other hand we show that it strictly corresponds to the longest chain of alternating primes hidden in the algebra.
Then, we start with any sequence
$p_1 \neq p_2 \neq \ldots \neq p_h$ of primes with $h\geq 2$ and construct
an example of the simplest algebra $\dhh$ that is $h$-step (but not $(h-1)$-step)
supernilpotent to demonstrate how to construct an $n$-ary conjunction polynomial $\pand_n$
of size $O(2^{c n^{1/(h-1)}})$.
This together with the assumption of Exponential Time Hypothesis
will be used to show the following theorem

\begin{thm}
\label{thm-lower}
The complexity for $\csat{\dhh}$ and $\ceqv{\dhh}$
is at least $\Omega(2^{c\cdot \log^{h-1}\cs})$,
where $\cs$ is the size of a circuit $\Gamma$ on the input
(unless ETH fails).
\end{thm}

Obviously these lower bounds would be even higher if one finds shorter conjunction terms.
Thus an upper bound for the complexity of $\csat{}$ relies on the (necessary) lenght of polynomials that are able to express $\pand_n$.
A kind of such lower bound had been already introduced as a conjecture
by Barrington, Straubing and Th\'{e}rien \cite{bst90}
in their study of non-uniform automata over groups.
To reword their conjecture for our purposes recall that a counting gate $\cmod^R_m$
(with unbounded fan-in) returns $1$ if all the $1$'s on the input sum up modulo $m$ to an element in $R$, and $0$ otherwise.
Moreover recall that $CC[m]$-circuit is build up with $\cmod^R_m$-gates only.
In this language the conjecture says that:
\begin{itemize}
  \item {\em the sizes of $CC[m]$-circuits $(\Gamma_n)_n$ with bounded depth
        computing $(\pand_n)_n$ grow exponentially in $n$.}
\end{itemize}
Very recently Kompatscher \cite{komp:nil19} has used this conjecture to show that
\begin{itemize}
  \item {\em for every nilpotent algebra $\m A$ from a congruence modular variety
        $\csat{\m A}$ and $\ceqv{\m A}$ can be solved in quasi polynomial time
        $O(2^{c\log^{t}m})$, for some constants $c,t$ depending on $\m A$.}
\end{itemize}
In our study of nilpotent algebras we had noticed that the exponent $t$ in Kompatscher's bound
is strongly correlated with $h$-step supernilpotency
(or in other words the depth of corresponding circuits).
This is now confirmed by Theorem \ref{thm-lower}, so that the above hypothesis has to be weakened accordingly.

A promising weaker version of such a hypothesis might be:
\begin{itemize}
  \item {\em the sizes of \ $CC[m]$-circuits $(\Gamma_n)_n$, of depth bounded by $h>1$,
        that compute $(\pand_n)_n$, grow at least as $\Omega(2^{cn^{1/(h-1)}})$.}
\end{itemize}
A dual version of the above hypothesis,
namely that to build $\cmod_m$ gates (of arbitrary large arity $n$) by circuits of bounded depth
needs superpolynomial (in $n$) number of classical Boolean gates $\pand, \por$ (of unbounded fan in) and $\pnot$, has been used by Furst, Sax and Sipser \cite{fss,sip} to seperate \pspace from polynomial hierarchy by oracles.
Later on Yao \cite{yao} confirmed this dual hypothesis,
while H{\aa}stad \cite{phdhas} has shown that for $\cmod_2$-gates the required sizes are even exponential.

\medskip
Unfortunately the hypothesis that $CC[m]$-circuits of depth $h>1$
require $\Omega(2^{cn^{1/(h-1)}})$ gates to express $(\pand_n)_n$
is blocked by Barrington, Beigel and Rudich in \cite{bbr94}.
They use integers $m$ that have $r\geq 2$ different prime factors
to construct $CC[m]$-circuits of depth $3$
that compute $(\pand_n)_n$ using only $2^{O(n^{1/r}\log n)}$ gates.
Such relatively small circuits are possible to built
by exploring the interaction of $r$ different primes on the very same level of the circuits.
However in our setting of the algebras $\dhh$ there is only one prime
at each level so that it suffices to use $CC[p_1,\ldots,p_h]$-circuits,
i.e. $CC$-circuits where on the $i$-th level there are only $\cmod_{p_i}$ gates, with $p_i$ being prime.
Note here that by our definition $CC[p_1,\ldots,p_h]$-circuits have depth $h$.
Thus the hypothesis we will build our upper bounds is the following
Strong Exponential Size Hypothesis (SESH).
\begin{conjj}[SESH]
The sizes of \ \ $CC[p_1, \ldots, p_h]$-circuits \ $(\Gamma_n)_n$, of depth $h>1$,
that compute $(\pand_n)_n$, grow at least as $\Omega(2^{cn^{1/(h-1)}})$.
\end{conjj}

Now, with the help of SESH we can show the upper bound for our problems that almost matches the lower bound of Theorem \ref{thm-lower}.

\begin{thm}
\label{thm-duper}
There are deterministic algorithms solving $\csat{\dhh}$ and $\ceqv{\dhh}$
in $O(2^{c\log^{h} \cs})$ time,
where $\cs$ is the size of a circuit $\Gamma$ on the input
(unless SESH fails).
\end{thm}

Relaxing deterministic realm to a probabilistic one
we can match the lower bound much better.

\begin{thm}
\label{thm-puper}
There  are  probabilistic  algorithms \ solving   $\csat{\dhh}$ \ \ and \ \ \ $\ceqv{\dhh}$
\ in \ time \ $O(2^{c\log^{h-1} \cs})$,
where $\cs$ is the size of a circuit $\Gamma$ on the input
(unless SESH fails).
\end{thm}

Note here that all the above theorems give  interesting  bounds only for $h\geq3$.
%Theorems \ref{thm-lower} and \ref{thm:duper}
In fact \cite{ikk:mfcs18} gives a polynomial upper bound for $\dpq$ with $p\neq q$.
However we decided to keep $h=2$ in our theorems,
as their proofs give a nice insight into the structure of the corresponding algebras.

Unfortunately this inside is still not deep enough to be generalized
to $2$-step supernilpotent algebras.
Both $\csat{}$ and $\ceqv{}$ remain open here.

On the other hand with $h\geq 3$ our results show that the dichotomy conjecture
(similar to the one for Constraint Satisfaction Problem) is unlikely to be confirmed.
It may even happen that $\csat{}$ for these algebras would provide a natural example
of an intermediate problem.

Our choice of the algebras $\dhh$ has been done very carefully so that
the proof of the above lower and upper bounds demonstrate the main general idea
but are still readable enough.
The important ingredient here is that the primes involved do alternate,
i.e. $p_1 \neq p_2 \neq \ldots \neq p_h$.
This corresponds to the fact that all the $\pand_n$'s can be obtained from $MOD_m$ gates
only if $m$ is not the power of a prime.

We decided to stay with our argument for this particular family of algebras
although most of the ideas used here can be generalized to $h$-step supernilpotent realm.
In fact in Section \ref{sect-strat} we show why alternation of primes is important
in the study of the expressive power of definable polynomials.
Unfortunately the arguments in a general setting have to be terribly involved
and make a heavy use of tame congruence theory \cite{hobby-mck} and modular commutator theory
\cite{freese-mck}.
A reader that is not experienced enough with universal algebraic tools may skip Section
\ref{sect-strat} and go directly to Section \ref{sect-example} where the main results are shown.
In Section \ref{sect-s4} we apply our methods to the special case of the symmetric group $\m S_4$
(but considered in its pure group language, without a possibility of endowing it by definable operations).
We do that as for many years the complexity of equation solving over this group has been unsettled.
In view of the fact that this problem is polynomial time for $\m S_3$
(while with \npc for $\csat{\m S_3}$)
there has been a hope that the same holds for $\m S_4$.
Now, under the assumption of ETH we destroy this hope.

\section{Stratification of algebras\label{sect-strat}}

Our study of $\csat{}$ for nilpotent algebras relies on the observation that there is a very strong connection between the depth $h$ of the $CC^0$-circuits and stratification of such algebras into $h$-step supernilpotent slices.
To define this stratification we start with recalling the concept of commutator.

If $\alpha$, $\beta$, $\gamma$ are congruences of an algebra then we say that {\em $\alpha$ centralizes $\beta$ modulo $\gamma$}, denoted
$C(\alpha,\beta;\gamma)$, if for every $n \geq 1$, every $(n+1)$-ary
term $\te t$, every $(a,b) \in\alpha$, and every
$(c_1,d_1),\dots,(c_n,d_n)\in\beta$ we have
\[
\te t(a,\bc c) \congruent{\gamma} \te t(a,\bc d)  \mbox{\ \  iff \ \  }
\te t(b,\bc c) \congruent{\gamma} \te t(b,\bc d).
\]
Obviously among all congruences $\gamma$ such that  $C(\alpha,\beta;\gamma)$ there is the smallest one and it is denoted by $[\alpha,\beta]$ and called the commutator of $\alpha$ and $\beta$.

By means of the commutator it is possible to define notions of
abelianity, solvability and nilpotency for arbitrary algebras.
First, for a congruence $\theta$ and $i=1,2,\dots$ we put
\[ \begin{array}{rclcrcl}
\theta^{(0)}&=&\theta & \ \ \ & \theta^{[0]}&=&\theta \\
\theta^{(i+1)}&=&[\theta,\theta^{(i)}] & \ \ \ &\theta^{[i+1]}&=&[\theta^{[i]},\theta^{[i]}].
\end{array}
\]

Now, a congruence $\theta$ of $\m A$ is called
{\em $k$-nilpotent} [or {\em $k$-solvable}]
if $\theta^{(k)}=0_{\m A}$ [$\theta^{[k]}=0_{\m A}$] and the algebra $\m A$
is {\em nilpotent} [{\em solvable}] if $1_A$ is $k$-nilpotent [$k$-solvable] for some finite $k$.
%In particular $\theta$ [or $\m A$] is abelian if $\theta^{(2)}=\theta^{[2]}=0_{\m A}$ [or $1_{\m A}^{(2)}=0_{\m A}$].

A more detailed discussions of the generalized commutator may be found in
\cite{freese-mck,hobby-mck,mmt}.%, \cite[Section 4.13]{mmt} and in \cite[Chapter 3]{hobby-mck}.

The concept of centrality and of the binary commutator has a natural generalization.
Namely, for a bunch of congruences
$\alpha_1,\ldots,\alpha_k,\beta,\gamma \in \con A$
we say that $\alpha_1,\ldots,\alpha_k$ centralize $\beta$ modulo $\gamma$,
and write $C(\alpha_1,\ldots,\alpha_k,\beta;\gamma)$,
if for all polynomials $\po f \in \pol A$ and all tuples
$\o a_1 \congruent{\alpha_1} \o b_1, \ldots, \o a_k \congruent{\alpha_k} \o b_k$
and $\o u \congruent{\beta} \o v$
such that
\[
\po f(\o x_1,\ldots, \o x_k, \o u) \congruent{\gamma} \po f(\o x_1,\ldots, \o x_k, \o v)
\]
for all possible choices of
$(\o x_1,\ldots, \o x_k)$ in $\set{\o a_1,\o b_1} \times \ldots \times \set{\o a_k,\o b_k}$
but $(\o b_1,\ldots.\o b_k)$,
we also have
\[
\po f(\o b_1,\ldots, \o b_k, \o u) \congruent{\gamma} \po f(\o b_1,\ldots, \o b_k, \o v).
\]
This notion was introduced by A.~Bulatov \cite{bulatov:supercomm}
and further developed by E.~Aichinger and N.~Mudrinski
\cite{aichinger-mudrinski}.
In particular they have shown that for all $\alpha_1,\ldots,\alpha_k \in \con A$
there is the smallest congruence $\gamma$ with $C(\alpha_1,\ldots,\alpha_k;\gamma)$
called the $k$-ary commutator and denoted by $\commm{\alpha_1}{\alpha_k}$.
Such generalized commutator behaves especially well in algebras from congruence modular varieties.
In particular this commutator is fully symmetric, monotone, join-distributive and we have
\[
\comm{\alpha_1} {\commm{\alpha_2}{\alpha_k}}
\leq \commm{\alpha_1}{\alpha_k}
\leq \commm{\alpha_2}{\alpha_k}
\]
We will often use this generalized commutator when some (or possibly all) of the $\alpha_i$'s coincide. Thus to emphasize the arity of this supercommutator we will sometimes write
$\commmk{\alpha_1}{\alpha_k}{k}$ instead of $\commm{\alpha_1}{\alpha_k}$.

We say that an algebra $\m A$ is $k$-supernilpotent if $\commmk{1}{1}{k}=0$.
The first inequality in the above display implies that a $k$-supernilpotent is $k$-nilpotent.

However, what is more interesting for our purposes,
is the going down with supernilpotent powers of the congruences in the fashion
of the solvable powers $\theta^{[i]}$.
Since in a finite algebra the sequence
\[
\theta \geq \comm{\theta}{\theta}=\comm{\theta}{\theta}_2 \geq \ldots \geq
\commmk{\theta}{\theta}{i} \geq \commmk{\theta}{\theta}{i+1}\geq\ldots
\]
has to stabilize,
the intersection $\scom{\theta}{1}=\bigcap_i \commmk{\theta}{\theta}{i}$ is actually one of the $\commmk{\theta}{\theta}{j}$'s.
Now we simply put
\[
\scom{\theta}{0}=\theta \mbox{\ \ \ and \ \ \ }
\scom{\theta}{k+1} = \bigcap_i \commmk{\scom{\theta}{k}}{\scom{\theta}{k}}{i}.
\]
This allows us to define $h$-step supernilpotent algebras, as those in which
$\scom{1}{h}=0$.
Note that $h$-solvable algebras are $h$-step supernilpotent in this sense,
so that $h$-step supernilpotent algebras need not be nilpotent.

However in our stratification we restrict ourselves to algebras that are nilpotent.
First we recall a very nice result (due to \cite{freese-mck} and \cite{kearnes:small-free})
illustrating the precise difference between nilpotency and supernilpotency
for finite algebras $\m A$ from congruence modular variety.
It says that the following two conditions are equivalent
\begin{itemize}
 \item $\m A$ is $k$-supernilpotent,
\item $\m A$ is $k$-nilpotent, decomposes into a direct product of algebras of prime power order
and the clone of all terms of $\m A$ is generated by finitely many operations.
\end{itemize}
This nice result can be localized.
To do that, first we need a concept of a characteristic of a covering pair $\theta\prec\delta$
of congruences (of type $\tn 2$, in the sense of Tame Congruence Theory \cite{hobby-mck}).
The fact that $\typ(\theta,\delta)=\tn 2$ says in particular that all traces of
$(\theta,\delta)$-minimal sets are, modulo $\theta$, (and up to polynomial equivalence) one-dimensional vector spaces over the same finite field.
The prime number that is the characteristic of this fields is also used to be called
the characteristic of the prime quotient $\theta\prec\delta$
and denoted by $\charr(\theta,\delta)$.
Now, for $\alpha\leq\beta$ we put
$\charrset{\alpha,\beta} =\set{\charr(\theta,\delta): \alpha\leq\theta\prec\delta\leq\beta}$.
Note here that in our setting if the intervals
$\intv{\alpha}{\beta}$ and $\intv{\alpha'}{\beta'}$
are projective
we have $\charrset{\alpha,\beta}=\charrset{\alpha',\beta'}$.
In case $\fj$ is a meet irreducible congruence, so that it has the unique cover, say $\fj^+$,
we will write $\charr(\fj)$ instead of $\charr(\fj,\fj^+)$.

The second concept needed to localize the characterization of supernilpotent algebras among the nilpotent ones, is a concept of a supernilpotent interval $\intv{\alpha}{\beta}$ in $\con{A}$.
We say that a congruence $\beta\geq\alpha$ is supernilpotent over $\alpha$
if $\scom{\beta}{1}\leq\alpha$.

The last concept needed is the one of a product interval.
We say that $\intv{\alpha}{\beta}$ is the product interval if there are congruences $\beta_1,\ldots,\beta_s$ (called decomposition congruences)
that intersect to $\alpha$ and for each $j$ satisfy
$(\bigcap_{i\neq j} \beta_i) \join \beta_j =\beta$.
In nilpotent algebras, if $\beta$ is supernilpotent over $\alpha$ then the interval $\intv{\alpha}{\beta}$ is not only a product interval,
but in fact it is prime uniform over each the $\beta_j$'s.
To be more precise, by a pupi (or prime uniform product interval) we mean
a product interval $\intv{\alpha}{\beta}$ in which each $\charrset{\beta_j,\beta}$
consists of a single prime, say $p_j$.
Finally we say that $\intv{\alpha}{\beta}$ is prime strongly uniform product interval
(psupi, for short)
if it is a pupi and moreover the primes $p_j$'s are different for different $\beta_j$'

Note here that modularity of the congruence lattice implies that
in a psupi there are no skew congruences between the $\beta_j$'s,
i.e. for every $\theta \in \intv{\alpha}{\beta}$ we have
$\theta = \bigcap_{j=1}^s (\theta\join\beta_j)$.
In particular each such $\theta$ that is locally meet irreducible
(i.e. meet irreducible in the interval $\intv{\alpha}{\beta}$)
has to lie over one of the $\beta_j$'s.

\medskip
Now, our localization says that for congruences $\alpha<\beta$
of  nilpotent algebra of finite type
(from congruence modular varieties)
the following two conditions are equivalent
\begin{itemize}
\item $\beta$ is supernilpotent over $\alpha$,
\item the interval $\intv{\alpha}{\beta}$ is prime uniform product interval.
\end{itemize}
A more detailed study of supernilpotent intervals
and $h$-step supernilpotent stratification
is contained in \cite{ikk:strat}.
Here we only note that the equivalence of the above conditions
can be shown by applying the VanderWerf's idea \cite{vanderwerf:phd} of wreath decomposition.
In fact this has been independently done by Mayr and Szendrei in \cite{mayr-szendrei}.

\medskip
For a better understanding of $h$-step supernilpotent algebras we observe first
that in a finite algebra $\m A$ for every congruence $\alpha$
there is the largest supernilpotent congruence over $\alpha$.
This is due to the fact that the join of two supernilpotent (over $\alpha$)
congruences $\beta_1,\beta_2$ is supernilpotent over $\alpha$.
Indeed, this supernilpotency can be witnessed by
$\commmk{\beta_i}{\beta_i}{k} \leq \alpha$ with the same $k$ for both the $\beta_i$'s.
But now taking $\commmk{\beta_1\join\beta_2}{\beta_1\join\beta_2}{2k}$
and distributing over the join
we get a join of $2k$-folds supercommutators of the $\beta_i$'s.
Since in each such supercommutator one of the $\beta_i$'s occurs at least $k$-times,
this puts each of them, and therefore entire join of $2^{2k}$ summands, below $\alpha$.

This allows us to define the sequence
\[
0=\sigma_0 \leq \sigma_1 \leq \ldots \leq \sigma_k \leq \sigma_{k+1}\leq\ldots
\]
of congruences such that
$\sigma_{k+1}$ is the largest congruence that is supernilpotent over $\sigma_{k}$.
This sequence of supernilpotent intervals
strongly corresponds to the other one used to define $h$-step supernilpotency, namely
\[
\ldots \leq \scom{1}{k+1} \leq \scom{1}{k}\leq \ldots \leq \scom{1}{2}\leq \scom{1}{1}\leq \scom{1}{0}=1.
\]
Indeed, in $h$-step supernilpotent algebra,
we induct on $k=0,\ldots,h$ to show that $\scom{1}{h-k}\leq \sigma_k$.
To pass from $k$ to $k+1$ we start with distributing over the join in the first supercommutator
\begin{align*}
\scom{(\scom{1}{h-(k+1)}\join\sigma_k)}{1}
&\leq\scom{(\scom{1}{h-(k+1)})}{1} \join \scom{(\sigma_k)}{1}\\
&=\scom{1}{h-k} \join \scom{(\sigma_k)}{1}
\leq \sigma_k,
\end{align*}
where the last inequality follows by induction hypothesis.
But what we get means that
$\scom{1}{h-(k+1)}\join\sigma_k$ is supernilpotent over $\sigma_k$,
so that it has to be below $\sigma_{k+1}$, as required.

In particular, with $k=h$ we get that $\scom{1}{h}=0$ implies $\sigma_h=1$.
In a similar fashion one shows that if $\sigma_h=1$ then $\scom{1}{k}\leq\sigma_{h-k}$,
so that $\scom{1}{h}=0$.

This gives that a finite algebra is $h$-step supernilpotent (i.e. $\scom{1}{h}=0$) iff $\sigma_h=1$.

\medskip

%%%%%%%%%%%%%%%%%%%%%%%%%%%%%%%%%%%%%%%%%%
%%%%%%%%%%%%%%%%%%%%%%%%%%%%%%%%%%%%%%%%%%
After all this preparation we are ready to show how the alternation of primes
(crucial in our study of $\csat{}$ and $\ceqv{}$ for nilpotent algebras)
is connected with $h$-step nilpotency.

\begin{thm}
\label{thm-alt}
For a finite nilpotent algebra $\m A$ from a congruence modular variety the following two conditions are equivalent:
\begin{itemize}
  \item $\m A$ is $h$-step supernilpotent,
  \item every chain $\fj_1<\fj_2<\ldots<\fj_s$ of meet irreducible congruences
  with alternating characteristics
  (i.e. $\charr(\fj_i)\neq\charr(\fj_{i+1})$, for $i=1,\ldots,s-1$),
  has its lenght $s$ bounded by $h$.
\end{itemize}
\end{thm}

\begin{proof}
Suppose first that $\fj_1<\fj_2<\ldots<\fj_s$ is such an alternating chain of meet irreducible congruences in an $h$-step supernilpotent algebra.
The idea is to project them into the prime strongly uniform product intervals
of the form $\Sigma_i=\intv{\sigma_i}{\sigma_{i+1}}$
by sending $\theta$ to $f_i(\theta)=(\theta\meet\sigma_{i+1})\join\sigma_i$.
For meet irreducible $\theta$ with the unique cover $\theta^+$
pick $j$ to be maximal with $\sigma_j\leq\theta$.
Observe that then $f_{j}(\theta) < f_j(\theta^+)$,
as otherwise the congruences $\theta^+\meet\sigma_{j+1},\theta,\theta^+,\sigma_{j+1},\theta\join\sigma_{j+1}$
would form a pentagon.
Moreover one can show that for this particular $j$ the congruence $f_{j}(\theta)$ is
locally (i.e. in $\Sigma_j$) meet irreducible.

Thus, if $s>h$ then after projecting the $\fj_t$'s from our alternating chain
into the $\Sigma_i$'s
at least two consecutive ones will fall into the same psupi, say $\Sigma_j$,
without collapsing them with their covers.
But after such projection, both of them are locally meet irreducible in $\Sigma_j$
so that being comparable they have to be over the same decomposition congruence $\beta_i$.
Consequently they must have the same characteristic, contrary to our alternating assumption.
This puts the bound for the alternating chain, as required.

%%%%%%%%%%%%%%%%%%%%%%%%
\medskip

Conversely, first note that since $\m A$ is nilpotent it is $h$-step supernilpotent for some $h$.
But now we will use intervals $\Pi_j=\intv{\scom{1}{j}}{\scom{1}{j-1}}$.
From the assumption that $\scom{1}{h-1}>0$ (i.e. $\Pi_h\neq\emptyset$)
we will construct the required chain of meet irreducible congruences of length $h$.

First, starting with an arbitrary prime $p_{h}\in\charrset{\Pi_{h}}$ we
can go down with $j=h-1,\ldots,1$ to isolate the consecutive primes $p_j \in \charrset{\Pi_j}$
satisfying $p_{j} \neq p_{j+1}$.
After fixing $p_{h},\ldots,p_{j+1}$
the possibility to properly choose $p_{j}$ is equivalent to
$\charrset{\Pi_{j}}\not\ci\set{p_{j+1}}$.
This obviously holds if $\card{\charrset{\Pi_{j}}}\geq 2$
so that we assume, to the contrary, that $\charrset{\Pi_{j}}=\set{p_{j+1}}$.
But $p_{j+1}\in \charrset{\Pi_{j+1}}$ is actually witnesses by one of the decomposition congruences,
say $\beta$, of the interval $\Pi_{j+1}$ so that $\charrset{\beta,\scom{1}{j}}=\set{p_{j+1}}$.
But this gives $\charrset{\beta,\scom{1}{j-1}}=\set{p_{j+1}}$.
Indeed, by modularity,
every covering pair $\beta\leq \theta\prec\delta\leq\scom{1}{j-1}$
either projects down to
$\scom{1}{j}\meet\theta \prec \scom{1}{j}\meet\delta$ inside $\intv{\beta}{\scom{1}{j}}$,
or up to $\scom{1}{j}\join\theta \prec \scom{1}{j}\join\delta$ inside $\Pi_{j+1}$.
In either case it inherits the characteristic $p_{j+1}$.
However now, $\charrset{\beta,\scom{1}{j-1}}=\set{p_{j+1}}$ yields that
$\scom{1}{j-1}$ is supernilpotent over $\beta$ so that we get a contradiction
$\scom{1}{j}\leq\beta$.

\medskip

Now, knowing that there is an alternating chain of primes from $\Pi_h\times\ldots\times\Pi_1$
we will inductively show that $\m A$ has $h$-long chain of meet irreducible congruences
$\psi_h <\psi_{h-1}<\ldots<\psi_1$ with alternating characteristics.
The characteristics of these meet irreducibles do not necessarily coincide with the one from the starting chain of primes, as during the recursion process we call our procedure for a smaller (quotient) algebra $\m A'$ in which the sets $\Pi'_j$ may be smaller.
Thus, when passing from $\m A$ to $\m A'$ the initial sequence
$p_h,p_{h-1},\ldots,p_1$ may change to $p_h,p'_{h-1},\ldots,p'_1$,
but the prime $p_h$ at the lower level remains unchanged.
All we need to take care of is that $p_h\neq p'_{h-1}$.

The easier case is when $\scom{1}{h-1}$ does not cover $0$
so that we can pick $0<\alpha\prec\scom{1}{h-1}$.
Passing to the quotient algebra $\m A' = \m A/\alpha$ we know that its chain of intervals $\Pi'_j$
coincide with the original one of the $\Pi_j$,
except $\Pi'_h=\intv{\alpha}{\scom{1}{h-1}}$.
But constructing the chain of the $p_j$'s for $\m A$
we may start with $p_h=\charr(\alpha,\scom{1}{h-1})$.
Then the chain of meet irreducibles for $\m A'$ nicely serves also for the original $\m A$.

Also, if $\scom{1}{h-1}$ is the unique atom of $\m A$, the algebra is subdirectly irreducible,
i.e. $0$ is a meet irreducible congruence of $\m A$.
Induction hypothesis applied to the quotient $\m A=\m A/\scom{1}{h-1}$,
but this time with $h$ smaller by $1$ and the shorter chain of primes $p_{h-1},\ldots,p_1$
obtained from the one for $\m A$ by simply deleting $p_h$,
equip us with the $(h-1)$-long chain of meet irreducibles, which after adding $\psi_h=0_A$
serves pretty well for $\m A$.

In the last case we have two different atoms $\scom{1}{h-1}$ and $\alpha$ in $\con{\m A}$.
Again we will pass to the quotient $\m A'=\m A/\alpha$,
but this time to make sure that this is going to work
we need to make sure that the new intervals
$\Pi'_j=\intv{\scom{1}{j}\join\alpha}{\scom{1}{j-1}\join\alpha}$'s
are non trivial (so that exactly $h$ corresponding primes can be chosen at all).
Suppose to the contrary that for some $j<h$ we have
$\alpha\join\scom{1}{j}=\alpha\join\scom{1}{j-1}$, so that
$\scom{1}{j-1}\leq\alpha\join\scom{1}{j}$.
Obviously $\alpha\not\leqslant\scom{1}{j}$,
as otherwise $\scom{1}{j-1}\leq\scom{1}{j}$,
contrary to our assumption that the seqence of the $\scom{1}{j}$'s is strictly decreasing.
On the other hand $\alpha\leq\scom{1}{j-1}$,
as otherwise $\scom{1}{j-1}$ and $\alpha$ would meet to $0$ and therefore
(together with $\scom{1}{j}$) would generate a pentagon.
In fact $\alpha\neq\scom{1}{h}$ tells us that then $\alpha<\scom{1}{j-1}$,
so that we can pick $\gamma$ with $\alpha\leq\gamma\prec\scom{1}{j-1}$.
As every congruence is supernilpotent over each of its subcovers,
we get that $\scom{1}{j}\leq\gamma$
and consequently we get a contradiction
$\gamma \geq\scom{1}{j}\join\alpha=\scom{1}{j-1}\join\alpha =\scom{1}{j-1}$.

Now note that although the intervals $\Pi_j$'s may loose the prime $\charr(0,\alpha)$,
we know that the only candidate for $p_h$,
namely $\charr(0,\scom{1}{h-1})$
still stays in $\charrset{\Pi'_j}$ as $p_h=\charr(0,\scom{1}{h-1})=\charr(\alpha,\alpha\join\scom{1}{h-1})$.
\end{proof}

\section{A paradigm for $h$-step supernilpotent algebras \label{sect-example}}

We start with an algebra that will serve us as paradigm for our considerations.
Fix a positive integer $h$ and a sequence $p_1,p_2,\ldots,p_h$ of primes.
Define an algebra $\dhh$ to be the expansion of the product
$\m Z_{p_1} \times \ldots \times \m Z_{p_h}$
of Abelian groups by the additional unary operations
$e_1,\ldots,e_h$ and $v_1,\ldots, v_{h-1}$
defined for $x=(x_1,\ldots,x_h) \in \m Z_{p_1} \times \ldots \times \m Z_{p_h}$ by
\begin{eqnarray*}
e_j(x)   &=& (0,\ldots,0,x_j,0,\ldots,0),\\
v_j(x)   &=& (0,\ldots,0,b_j(x_{j+1}),0,\ldots,0),
\end{eqnarray*}
where $b_j : \m Z_{p_{j+1}} \map \m Z_{p_{j}}$ is a function given by
$b_j(0)=0$ and $b_j(a)=1$ otherwise.

Note here that
\begin{itemize}
  \item the algebra $\dpp{p}$ is simply the group $\m Z_p$, so that it is Abelian,
  \item the algebra $\dpp{p,q}$, with $p\neq q$ had been extensively studied in \cite{ikk:mfcs18} where a polytime algorithm was presented both for $\csat{\dpp{p,q}}$ and $\ceqv{\dpp{p,q}}$.
\end{itemize}
Here we will study the algebras of the form $\dhh$ with the assumption that the sequence $p_1,p_2,\ldots,p_h$ of primes is alternating, i.e. $p_i \neq p_{i+1}$. Then we will show that
\begin{itemize}
  \item the algebra $\dhh$ is $h$-nilpotent (actually $h$-step supernilpotent).
\end{itemize}
Since the algebra $\dhh$ has an underlying group structure each equation of polynomials
$\po t = \po s$ that may be an input to $\csat{}$ or $\ceqv{}$
can be replaced by $\po t-\po s=0$ so that we restrict ourselves to the equations of this special shape.

\subsection{The structure of $\dhh$\label{sect-structure}}
To understand the algebra $\dd=\dhh$ we start with defining a couple of its constants
$0=(0,\ldots,0), 1=(1,\ldots,1)$
and polynomials $e^k$ by putting $e^k(x) =\sum_{j\geq k} e_j(x)$.

Now it is easy to observe that the relations
\[
\theta_k = \set{(a,b)\in \dd^2 : e^k(a)=e^k(b)}
\]
together with the total congruence $\theta_{h+1}$ form a chain
$0=\theta_1 < \theta_2 < \ldots <\theta_h < \theta_{h+1} =1$
and that they are actually all congruences of $\dd$
-- indeed every principal congruence of $\dd$ is one of the $\theta_i$'s.

Inducting on the complexity of a polynomial $\po t(x_1,\ldots,x_n)$ of $\dd$
we can easily show that $e_j\po t(\o a) = e_j\po t(\o b)$,
whenever $j\geq k$ and $a_i \theta_k b_i$.
This means that a polynomial having a range contained in
$e_k(D) =\set{0}\times\ldots\times\set{0}\times \m Z_{p_k}\times\set{0}\times\ldots\times\set{0}$
does not depend on the values of the first $k-1$ summands in $x=\sum_{j=1}^h e_j(x)$.
Also an inspection of the behavior of the basic operations of $\dd$
(in particular noticing that $e_k(x+y)=e_k(x)+e_k(y), e_k(e_k(x))=e_k(x), e_k(v_k(x))=v_k(x)=v_k(e_{k+1}(x))$ and $e_k(e_\ell(x))=0=e_k(v_\ell(x))$ for $k\neq\ell$),
allows us to represent every polynomial $\po t(\o x)$ with the range contained in $e_k(D)$,
i.e. a polynomial satisfying $\po t =e_k\po t$, by a sum of expressions of the form
$e_kx_i, e_kc$ or $v_k\po s$,
where $x_i$ is a variable, $c$ is a constant and $\po s$ is some polynomial of $\dd$.
In order to have $v_k\po s \neq 0$ we may assume that the range of $\po s$ is contained in
$e_{k+1}(D)$, as $v_k\po s = v_k e_{k+1}\po s$.
However, as we have already noticed, polynomials with the range contained in $e_{k+1}(D)$
depends only on the projections $e^{k+1}(x_i)$ of its variables.
Summing up we know that
\[
e_k\po t(\o x) =
c + \sum_{i=1}^n \lambda_i\cdot e_k(x_i) +
\sum_{\po s \in S} \kappa_{\po s}\cdot v_k\po s(e^{k+1}(x_1),\ldots,e^{k+1}(x_n)),
\]
where $c\in e_k(D)$ is a constant,
the multiplication by the scalars $\lambda_i$'s or $\kappa_{\po s}$
(taken from $\m Z_{p_k}$)
is a shortening for adding the appropriate elements appropriate number of times,
and $S$ is a set of polynomials of $\dd$ with ranges contained in $e_{k+1}(D)$.

To estimate the length of the above representation of $e_k\po t$ note that
since $e_k$ distributes over the addition,
we know that the number of summands in the above display
(including those hidden in the $\lambda_i$'s and the $\kappa_{\po s}$'s)
is bounded by the number of additions in $\po t$.
Moreover note that each $\po s \in S$ is in fact a subterm of $\po t$
(and that they are pairwise disjoint subterms of $\po t$)
so that $\sum_{\po s\in S} \card{\po s} \leq \card{\po t})$.
In particular the length of the above representation is bounded by $O(\card{\po t})$.
We will often refer to this representation as the canonical representation
keeping in mind that $\po t(\o x) = \sum_{j=1}^k e_k\po t(\o x)$ and that
\begin{equation}
\begin{split}
\label{canon}
e_h\po t(\o b)  &=
c^h + \sum_{i=1}^n \lambda^h_i\cdot e_h(b_i),
\\
e_{h-1}\po t(\o b) & =
c^{h-1} + \sum_{i=1}^n \lambda^{h-1}_i\cdot e_{h-1}(b_i) \\
        &\ \ \ \ +
\sum_{\po s \in S^{h-1}} \kappa_{\po s}^{h-1}\cdot v_{h-1}\po s(e^{h}(b_1),\ldots,e^{h}(b_n)),
\\
& \vdots
\\
e_1\po t(\o b) & =
c^{1} + \sum_{i=1}^n \lambda^{1}_i\cdot e_1(b_i) \\
      &\ \ \ \ +
\sum_{\po s \in S^{1}} \kappa_{\po s}^1\cdot v_{1}\po s(e^{2}(b_1),\ldots,e^{2}(b_n)).
\end{split}
\end{equation}
But what is more important for us is that such (relatively short) canonical representation
can be obtained not only from a polynomial of $\dd$ but also from a circuit $\Gamma$ over $\dd$ that computes this polynomial.
This is not entirely obvious, as sometimes circuits may have logarithmic size with respect to the length of a polynomial they compute.
Each node of the circuit $\Gamma$ determines a subcircuit $\Gamma'$ of $\Gamma$.
With each $\Gamma'$ we associate a polynomial $\po t_{\Gamma'}$ (possibly too large)
in such a way that $\po t_{\Gamma''}$ is a subpolynomial of $\po t_{\Gamma'}$
whenever $\Gamma''$ is determined by a node in $\Gamma'$.
Despite the sizes of the $\po t_{\Gamma'}$'s we go to their canonical representations,
as described in \eqref{canon}.
All the data we need to store for the $e_j\po t_{\Gamma'}$'s $(j=1,\ldots,h)$ are
the constants $c^j, \lambda^j_i, \kappa^j_{\po s}$ and the sets $S^j$ themselves.
There is an easy bound for the constants, once we bound $\card{S^j}$.
To unwind this recursive construction note that for each $\po s \in S^j$ we need to keep its data only for one level, namely $e_{j+1}\po s$, as $v_j\po s= v_je_{j+1}\po s$.
Now, since $\card{S^j}$ is bounded from above by the number of subcircuits of $\Gamma'$
we get $\card{S^j}\leq \card{\Gamma'}$.
So, unwiding this construction for entire $\Gamma$
we get that our canonical representation of $\po t_\Gamma$
is of size $O(\card{\Gamma}^{h^2})$.
Thus, in what follows, we will simply put our lower and upper bounds in terms of
the size of canonical representation for polynomials rather than for circuits.

\medskip

Our next observation shows a connection between some polynomials of $\dhh$ and $CC[p_1,\ldots,p_h]$-circuits.

\begin{fact}
\label{fact-gates}
For an $n$-ary  polynomial $\po g$ of $\dd$, and $j<k$ the mapping 
$$v_je_{j+1} \po g (e_k x_1,\ldots,e_k x_n) : \set{0, e_k1}^n \map \set{0,e_j1}$$
can be simulated by a $CC[p_{j+1},\dots,p_k]$-circuit $\Gamma$ 
of depth $k-j$ and size $O(\card{\po g})$ in a way, that
\[
v_je_{j+1}\po g (e_k x_1,\ldots,e_k x_n) =
\left\{
\begin{array}{ll}
0,     &\mbox{if \ } \Gamma(\bb(x_1),\ldots,\bb(x_n))= 0,\\
e_j1,  &\mbox{if \ } \Gamma(\bb(x_1),\ldots,\bb(x_n))= 1,
\end{array}
\right.
\]
where the Boolean function $\bb(x): \set{0, e_k1} \map \set{0,1}$
returns $0$ if $e_k(x)=0$ and $1$ otherwise.
\end{fact}

\begin{proof}
We induct on $j=k-1,\ldots,1$ to built the required circuit from the gates $MOD^R_{p_i}$'s.
Since for $j=k-1$ the function $e_k\po g e_k$ actually maps $e_kD^n$ into $e_kD$,
it has to be an affine function of the form $c+\sum_{i=1}^n \lambda_i e_kx_i$.
Thus $v_{k-1}e_k\po g e_k$ can be simulated by one gate $MOD_{p_k}^{\m Z_{p_k}-\set{-c}}$
with each of the $\bb(x_i)$'s put to the gate $\lambda_i$ times on input.

Going down with $j$ our canonical form gives that
\begin{equation*}
\begin{split}
e_{j+1}\po g(e_k\o x) &=
c^{j+1} + \sum_{i=1}^n \lambda^{j+1}_i\cdot e_{j+1}(e_kx_i) \\
&+\sum_{\po s \in S^{j+1}} \kappa^{j+1}_{\po s}\cdot v_{j+1}\po s(e^{j+1}(e_kx_1),\ldots,e^{j+1}(e_kx_n)),
\end{split}
\end{equation*}
which actually reduces to
\[
e_{j+1}\po g(e_k\o x)  =
c^{j+1} + \sum_{\po s \in S^{j+1}} \kappa^{j+1}_{\po s}\cdot v_{j+1}\po s(e_kx_1,\ldots,e_kx_n),
\]
as $e_{j+1}e_k x=0$ and $e^{j+1}e_kx=e_kx$.
Now, given the circuits $\Gamma_{\po s}$
that do the job for all the $v_{j+1}e_{j+2}\po s e_k$ with $\po s \in S^{j+1}$,
we feed $MOD_{p_{j+1}}^{\m Z_{p_k}-\set{-c^{j+1}}}$
with each $\Gamma_{\po s}$ repeated $\kappa^{j+1}_{\po s}$ times.
\end{proof}

\medskip

Using our understanding of polynomials of $\dd$,
provided by the canonical representation \eqref{canon},
we can now easily determine the behavior of the commutator of congruences of $\dd$.
Namely if $i\leq j$ then $\com{\theta_i}{\theta_j} = \theta_{i-1}$.

We start here with an adaptation of Lemma 3.1 from \cite{ikk:mfcs18}.
The original Lemma has been formulated for the algebra of the form $\dpq$ with $p\neq q$,
while we will need it in our more general context of $\dhh$.
Obviously the algebra $\dpp{p_k,p_{k+1}}$
can be identified with
$e_k\dd + e_{k+1}\dd =
\set{0} \times \ldots \times \set{0} \times
\m Z_{p_k} \times \m Z_{p_{k+1}} \times \set{0} \times \ldots \times \set{0} \ci \dd$.

\begin{lm}\label{lemma31}
For \ $k<\ell$ and all \ $m$,
every function of the form $g : (e_\ell D)^m \map e_k D$
can be represented by an $m$-ary  polynomial $\po p$ of $\m D$,
with both its size and the time needed to actually compute it bounded by $O(2^{cm})$,
where the constant $c$ depends only on the algebra $\dd$.
\end{lm}

\begin{proof}
We start with observing that for a polynomial $v'_k(x) = e_k1-v_k(e_{k+1}1-e_{k+1}x)$
we have $v'_k(0)=0, v'_k(e_{k+1}1)=e_k1\neq 0$
and $v'_k(e_{k+1}a)=0$ for all $a\in e_{k+1}D -\set{e_{k+1}1}$.
Thus $\sum_{a\in e_{k+1}D -\set{0}} v'_k(a) = v'_k(e_{k+1}1)\neq 0$,
so that we are in the scope of Lemma 3.1 of \cite{ikk:mfcs18}
which yields a required polynomial $\po p$ representing the function $g$.
Moreover the shape of this polynomial (provided in that Lemma) allows us to bound its size
(and the time to produce it) by
$O(p_{k+1}^{m+1}\cdot p_k \cdot m \cdot (p_{k+1}+1)) = O(p_{k+1}^{m}\cdot m)$, as required.

Now, if $\ell>k+1$, we inflate each variable $x_i, (i=1,\dots,m)$  into $p_\ell$ variables
$x_i^0,\ldots,x_i^{p_\ell-1}$
This allows us to project an element $a\in e_\ell D$ into a tuple
$(a^0,a^1,\ldots,a^{p_\ell-1})\in e_kD^{p_\ell}$
by putting $a^\lambda = 1-v_{k+1}\ldots v_{\ell-1}(a-\lambda\cdot e_\ell1)$.
Note that for each such $a$ exactly one of the $a^\lambda$'s is nonzero (actually it is $e_{k+1}1$),
namely the one with $\lambda$ occurring in the $\ell$-th position of $a=(a_1,\ldots,a_h)$.
Take any function $g':(e_{k+1}D)^{p_\ell\cdot m} \map e_kD$
satisfying
\[
g'(x^0_1,\ldots,x^{p_\ell-1}_1,x^0_2,\ldots,x^0_m,\ldots,x^{p_\ell-1}_m)=g(x_1,\ldots,x_m)
\]
whenever
$x^\lambda_i = 1-v_{k+1}\ldots v_{\ell-1}(x_i-\lambda\cdot e_\ell1)$.
Using the case $\ell=k+1$ the function $g'$ can be represented by a $p_\ell\cdot m$-ary polynomial $\po p'$ od $\dd$.
It should be obvious that now substituting $1-v_{k+1}\ldots v_{\ell-1}(a-\lambda\cdot e_\ell1)$'s
for the $x^\lambda_i$'s we get an $m$-ary polynomial $\po p$ of $\dd$ representing $g$.
Moreover $\card{\po p} \leq O(\card{\po p'}) \leq O(2^{c'm})$ with $c'=p_\ell c$.

An inspection of the proof of Lemma 3.1 in \cite{ikk:mfcs18} provides a bound for the time needed to actually find the required polynomials, as claimed.
\end{proof}

\bigskip
We conclude this subsection with mentioning a very nice feature of the algebra $\dd$.
Namely $\dd$ is as rich in polynomials as possible.
This means that every function $g: D^n \map D$ that preserves congruences of $\dd$ and their commutator is already a polynomial of $\dd$.
As we are not going to use this fact in our future considerations we provide only a brief sketch of its proof.

Starting with $g$ preserving congruences and their commutator
we know that the algebra $\dd$ endowed by $g$ is still nilpotent.
Obviously $g$ can be represented as the sum $\sum_{k=1}^h e_kg$.
Since the range of $e_kg$ is contained in $e_k D$,
we can recursively apply Proposition 7.1 of \cite{freese-mck} to claim that
$e_kg(\o x)$ can be represented by
%$e_kg(\o x) =
$c^{k} + \sum_{i=1}^n \lambda^{k}_i\cdot e_k(x_i) + g'(e^{k+1}(x_1),\ldots,e^{k+1}(x_n))$,
where $g'$ is some $n$-ary function mapping $e^{k+1}D$ into $e_kD$.
Now, with a little bit more effort we can strenghten Lemma \ref{lemma31}
to represent every functions mapping simultaneously all upper levels
$e_{k+1}D,\ldots,e_hD$ (i.e. the entire $e^{k+1}D$ not just one level $e_\ell$ as in that Lemma)
into a lower level $e_kD$ by a polynomial of $\dd$. This would show that $g'$ and therefore $g$ are the polynomials.

\subsection{Lower bound\label{sect-lower}}
We are going to show that under the assumption of the ETH
%Exponential Time Hypothesis
the complexity for both $\csat{}$ and $\ceqv{}$ for $\dhh$
is at least $\Omega(2^{c\cdot \log^{h-1}\cs})$
where $\cs$ is the size of a circuit $\Gamma$ on the input.

To deal with $\csat{}$, for every formula $\Phi(x_1,\ldots,x_n)$ in 3-CNF
we will construct an $n$-ary polynomial $\po t_\Phi(x_1,\ldots,x_n)$
such that $\Phi$ is satisfiable iff the equation $\po t_\Phi(\o x)=e_11$ has a solution in $\dd$.
We will make sure that the time required to produce $\po t_\Phi$
is bounded by $O(2^{c m^{1/(h-1)}})$, where $m$ is the number of clauses in $\Phi$.
Now, having algorithms for $\csat{\dhh}$ working in time $O(2^{\epsi\cdot \log^{h-1}\cs})$
for arbitrary small $\epsi>0$
we would be able to solve 3-CNF-SAT in the very same time with $\cs$ replaced by
$d\cdot 2^{c m^{1/(h-1)}}$, i.e. in $O(2^{\epsi c^{h-1} m})$.
This obviously contradicts ETH (after remodelling it with the Sparsification Lemma).

%$2^{\sqrt[h-1]{\card{\po t}}}$

To produce $\po t_\Phi$ we start with the $s$-ary functions
$\pand^s_k : e_{k+1}D^s \map e_kD$
defined by $\pand^s_k(a_1,\ldots,a_s)=0$ if at least of the $a_i$'s is $0$,
and $\pand^s_k(a_1,\ldots,a_s)=e_k1$ otherwise.
Lemma \ref{lemma31} assures us that all those functions can be realized by polynomials of $\dd$ in $O(2^{cs})$ time, possibly with different constants $c$ depending on $k$.

Although the functions $\pand^s_k$ are long, the composition of two consecutive ones is shorter
(in terms of the variables involved).
Indeed the function
\[
\pand^s_{k-1}
\left(
\pand^s_k(x_1,\ldots,x_s),%\pand^s_k(x_{s+1},\ldots,x_{2s}),
\ldots,\pand^s_k(x_{(s-1)s+1},\ldots,x_{s^2})
\right)
\]
acts from $e_{k+1}D^{s^2}$ into $e_{k-1}D$ and can be produced in
$O(s\cdot 2^{c_ks c_{k-1}s}) = O(2^{cs})$ time.
Repeating this procedure we end up with a $s^{h-2}$-ary polynomial $\pand$,
of size/time $O(2^{cs})$, mapping $e_{h-1}D^{s^{h-2}}$ into $e_1D$
and behaving as a conjunction, i.e.
$\pand(\o a)=0$ if some of the $a_i$'s is $0$, and $\pand(\o a)=e_11$ otherwise.

The above part of our construction has been independent of $\Phi$.
We are going to use the action of the level $e_hD$ onto $e_{h-1}D$ to code $\Phi$.
To start with we define a boolean function $\bb : e_hD \map \set{\top,\bot}$
by putting $\bb(a)=\top$ for all $a\neq0$ and $\bb(0)=\bot$.
Now, if $m$ is the number of clauses in $\Phi$
we fix $s$ to be $\lceil m^{1/(h-1)}\rceil$
and split the clauses into $s^{h-2}$ parts, say $\Phi_\ell$'s,
each of which containing at most $s$ clauses,
so that each $\Phi_\ell$ involves at most $n_\ell\leq 3s$ variables.
Again we refer to Lemma \ref{lemma31} to ensure that the function
$\pcnf_{\Phi_\ell} : e_hD^{n_\ell} \map e_{h-1}D$
given by
$\pcnf_{\Phi_\ell}(a_1,\ldots,a_{n_\ell})=0$ if $\Phi_\ell(\bb(a_1),\ldots,\bb(a_{n_\ell}))=\bot$
and $\pcnf_{\Phi_\ell}(a_1,\ldots,a_{n_\ell})=e_{h-1}1$ otherwise,
can be realized by a polynomial (again in time bounded by $O(2^{cs})$).
For simplicity we make sure that occurrence of each variable $x$ is replaced by $e_h(x)$.

Now, filling up our $s^{h-2}$-ary polynomial $\pand$ with $n_\ell$-ary polynomials $\pcnf_{\Phi_\ell}$'s we finally arrive at the polynomial $\po t_\Phi$.
Again we produced it in $O(2^{cs})$ time, for some (possibly new) constant $c$.
It should be obvious that $\po t_\Phi$ does the required job for us.

\medskip

To get a similar lower bound for  $\ceqv{\dhh}$ it suffices to notice that the constructed polynomial $\po t_\Phi$ takes only two values: $0$ and $e_11$.
In such a case $\ceqv{}$ and $\csat{}$ could be bisimulated.

\subsection{Deterministic upper bound\label{sect-duper}}
This subsection is devoted to analyze a solution space for an equation $\po t(\o x)=0$ over the algebra $\dd=\dhh$. This analysis is based on SESH (Strong Exponential Size Hypothesis).
This will lead to an algorithm that solves the equations (and therefore satisfiability of circuits)
over $\dd$ in subexponential time almost matching the lower bound from Subsection \ref{sect-lower}.

Recall here that in a superniloptent algebra $\m A$ an equation has a solution
if it has one which is almost constant, say equal to $0$,
i.e. the number of non-zero values for $x_1,\ldots,x_n$
is bounded by a constant depending only on the algebra $\m A$.
We are using the algebra $\dhh$ as a paradigm for $h$-step supernilpotent algebras
to show (under the assumption of SESH) that in such realm
if an equation $\po t(\o x)=0$ has a solution
then it has one which again is {\em almost} constant,
but this time {\em almost} means
that there are at most $O(\log^{h-1} \card{\po t})$ non-zero values.
Thus to check if $\po t(\o x)=0$ has a solution
it suffices to check if there is one among those {\em almost} constant tuples.
Since there are at most
$O(n^{\log^{h-1} \card{\po t}} \cdot \card{D}^{\log^{h-1} \card{\po t}})
= O(2^{c\log^{h} \card{\po t}})$ such candidates,
while checking if one is actually a solution takes roughly $O(\card{\po t})$,
we have an algorithm working in $O(2^{c\log^{h} \card{\po t}})$ time.

\medskip
For two tuples $\o a =(a_1,\ldots,a_n)$ and $\o b =(b_1,\ldots,b_n)$ from $D^n$
we put $\equa{\o a=\o b}=\set{i : a_i = b_i}$
and analogously for $\equa{\o a \neq\o b}$.

Now we will show how a solution $\o a$ of $\po t(\o x)=0$ can be successively modified,
by to get a sequence of solutions
\(
\o a= \o a^0 \rightarrow \o a^1 \rightarrow \ldots \rightarrow\o a^h.
\)
When passing from $\o a^{k-1}$ to $\o a^{k}$ we will introduce zeros on the $k$-th coordinate,
i.e. making $e_k(a_i)=0$, for more and more $i$'s, while keeping the other coordinates unchanged.
To be more precise we will make sure that
$\cardd{\equa{e_k(\o a^k) \neq 0}} \leq O(\log^{k-1}\card{\po t})$
and $e_j(\o a^k) = e_j(\o a ^{k-1})$ for $j\neq k$.
Thus we will get
\[
\cardd{\set{i: \sum_{j\leq k}e_j(a^k_i)\neq 0}} \leq
\sum_{j\leq k} O(\log^{j-1}\card{\po t}) = O(\log^{k-1}\card{\po t}),
\]
so that finally arriving at $\o a^h$ we end up with
$\cardd{\equa{\o a^h \neq 0}} \leq O(\log^{h-1}\card{\po t})$, as promised.

To keep our second invariant when passing from $\o a^{k-1}$ to $\o a^k$
we need to stay inside the set
\[
E^k =\set{\o b : e_j(\o b) = e_j(\o a^{k-1}) \mbox{\ for \ }j\neq k },
\]
in particular we secure $e^{k+1}(\o b)=e^{k+1}(\o a^{k-1})$
so that for $j>k$ we have $e_j\po t(\o b) = e_j\po t(\o a^{k-1})$.

In particular, when producing $\o a^1 \in E^1$ we need to take care only of
\(
e_1\po t(\o a^1) = c^1 +\sum_i \lambda^1_i e_1(a^1_i) + \po t'(e^2(\o a^1)).
\)
However $\po t'(e^2(\o b))$ gives the same value for all $\o b\in E^1$.
Thus our requirement that $\o a^1$ is still a solution reduces to the equation
$0=e_1\po t(\o a^1) = c' +  \sum_i \lambda^1_i e_1(a^1_i)$.
Therefore one can easily find a solution $\o a^1$ to this linear equation with at most one of the $a^1_i$'s being non-zero.

Also, when passing from $\o a^{k-1}$ to $\o a^k$ we choose $\o a^k \in E^k$ to be a solution to
$\po t(\o x)=0$ that maximizes the number of zeros for $e_k(a^k_1),\ldots,e_k(a^k_n)$.
If $\o a^k$ would still have too many non-zeros we will construct a relatively short polynomial
(of the arity corresponding to the number of those nonzeros) that behaves as conjunction and refer to SESH to get a contradiction.

We start this argument with a better understanding of solutions $\o b \in E^k$ to the equation
$\po t(\o x)=0$.
For such $\o b$ to be a solution reduces to the system of equations:
\begin{eqnarray*}
0&=&e_k\po t(\o b,)\\
0&=&e_{k-1}\po t(\o b),\\
& \vdots & \\
0&=&e_1\po t(\o b),
\end{eqnarray*}
where each $e_j\po t(\o b)$ is representen in its canonical form as in \eqref{canon}.
The last sum in the representation of $e_k\po t$ occurs only if $k\neq h$.
Actually this sum disappears independently of how big is $k$.
This is because this sum is constant on the set $E^k$.
Also the linear parts in all equations with $j<k$
are constant as $e_j(b_i)=e_j(a^{k-1}_i)$ for $\o b \in E^k$.
This also allows to replace $e^j(b_i)$ by $e_k(b_i)$.
By possibly modifying the constants $c^1,\ldots,c^k$
(and the sets $S^1,\ldots,S^{k-1}$ of polynomials)
we are left with finding $\o b \in E^k$ satisfying
\begin{equation}
\begin{aligned}\label{rj}
0&=& c^k &+ \sum_{i=1}^n \lambda^k_i\cdot e_k(b_i),
\\
0&=& c^{k-1} &+ \sum_{\po s \in S^{k-1}} \kappa_{\po s}^{k-1}\cdot v_{k-1}\po s(e_{k}(b_1),\ldots,e_{k}(b_n)),
\\
& \vdots &&
\\
0&=& c^{1} &+ \sum_{\po s \in S^{1}} \kappa_{\po s}^1\cdot v_{1}\po s(e_k(b_1),\ldots,e_k(b_n)).
\end{aligned}
\end{equation}
We want to replace this system of equations by a single equation (of about the same size).
We will do it with the help of the $(\card{S^1}+k-1)$-ary function
$\pv : e_2D^{\card{S^1}+k-1} \map e_1D$
defined (on the variables $z_{\po s}$ indexed by $\po s \in S^1$ and $z_2,\ldots,z_k$) by
\begin{equation*}
\begin{split}
\pv&(\ldots, z_{\po s},\ldots,z_2,\ldots,z_k) = \\
&\left(e_11-\left(c^1+\sum_{\po s \in S^1}
\kappa_{\po s}^1\cdot v_1(z_{\po s})\right)^{p_1-1}\right)\cdot \prod_{j=2}^{k} (e_11-v_1(z_j)).
\end{split}
\end{equation*}
Note that $\pv(\ldots, z_{\po s},\ldots,z_2,\ldots,z_k)=e_11$
iff all the $z_2,\ldots,z_k$ as well as $c^1+\sum_{\po s \in S^1} \kappa_{\po s}^1\cdot v_1(z_{\po s})$ are zeros.
Now, denoting by $\po r_j(\o b)$ the right hand side of the $j$-th equation (counting from the bottom) and substituting $v_2\ldots v_{j-1}r_j(\o b)$ for $z_j$
and $\po s(e_k(b_1),\ldots,e_k(b_n))$ for the $z_{\po s}$'s,
we reduced our system of equations to just one equation of the form $\pv(....)= e_11$,
where inside $\pv$ there are polynomials of $\dd$ with total length bounded by $O(\card{\po t})$.

Obviously, by Lemma \ref{lemma31}, $\pv$ can be represented by a polynomial of $\dd$.
However to have a control of its size we need a little bit more subtle argument.
First we distribute all the multiplications in $\pv$ to end up with a sum of a constant and expressions of the form $v_1(y_1)\cdot\ldots\cdot v_1(y_\ell)$, with $y_i$'s being the variables $z_j$'s or $z_{\po s}$'s.
It should be obvious that this sum has at most
$(1+\card{\po t}^{p_1-1}\cdot 2^{k-1})\leq O(\card{\po t}^{p_1})$ summands.
Moreover $\ell$ is bounded by a constant $(p_1-1)+(k-1)$ independent of $\po t$.
This allows us to call Lemma \ref{lemma31} to represent all the $\ell$-ary functions
$v_1(y_1)\cdot\ldots\cdot v_1(y_\ell)$ by polynomials of $\dd$
with lengths bounded by a constant independent of $\po t$.

Up to now, we end up with a polynomial $\po t^\star(\o x)$
of size $O(\card{\po t}^c)$ (for some constant $c$)
such that inside $E^k$ the equations $\po t(\o x)=0$ and $\po t^\star(\o x)=e_11$
have the same solutions.
Moreover, the shape of $\pv$ tells us that $\pv$ (and therefore $\po t^\star$)
takes only two values, namely $0$ and $e_11$
and therefore we will modify it to simulate the operation of conjunction
with entries from $\set{0,e_k1}$ and values $0, e_11$.
The fact that in the polynomials $\po r_j$ (and therefore in $\po t^\star$)
all variables $x_i$ are in the scope of $e_k$ will be helpful in our further analysis.

By our choice $\o a^k\in E^k$ is a solutions to $\po t^\star(\o x)=e_11$
minimizing the cardinality of the set $\equa{e_k(\o a^k)\neq 0}$.
Now we modify $\po t^\star$ to $\po t^\dstar$,
first by fixing each variable $x_i$ to be $a^k_i$ whenever $e_k(a^k_i)=0$
and then by replacing each of the remaining variables $x_i$ by $\lambda_i\cdot x_i$
where $\lambda_i$ is the unique nonzero coordinate of $e_k(a^k_i)$
(and as previously $\lambda \cdot x$ is the sum $x+\ldots+x$ with $\lambda$ summands).
Let $\ell$ be the arity of $\po t^\dstar$ so that without loss of generality we may assume that the first $\ell$ variables of $\po t^\star$ survived.
We claim that $\po t^\dstar$ is the required conjunction.
Indeed, $\po t^\dstar(e_k1,\ldots,e_k1)=\po t^\star(e_k(\o a^k))=t^\star(\o a^k)= e_11$,
while, by maximality of $\equa{e_k(\o a^k)=0}$,
a tuple $\o b\in D^\ell$ with $b_i=0$ for $i\leq \ell$
cannot be a solution to $\po t(\o x)=0$ so that $t^\dstar(\o b) =0$.

Now Fact \ref{fact-gates} allows us to create a circuit of size
$O(\card{\po t^\dstar})=O(\card{\po t}^d)$ and of depth $k$
that computes the $\ell$-ary conjunction.
However SESH tells us that the size of this circuit
has to be at least $\Omega(2^{c \ell^{1/(k-1)}})$.
This gives $\cardd{\equa{e_k(\o a^k)\neq 0}}=\ell\leq O(\log^{k-1}\card{\po t})$,
as required.

\bigskip
To see that $\ceqv{\dhh}$ can be solved roughly in the very same time,
note that determining if the identity $\po t(\o x)=0$ holds we need to check that none of the equations of the form $\po t(\o x)-d=0$, with $d\in D-\set{0}$ has a solution.

\subsection{Probabilistic upper bound\label{sect-puper}}
We present a randomized algorithm for checking whether an equation
$\po t(\o x)=0$ has a solution over $\dd$.
This time, again using SESH, we will show that if a polynomial $\po t$ returns some value $d\in D$,
i.e. $\po t^{-1}(d)\neq\emptyset$ then it actually returns this value many times, namely
\(
\card{\po t^{-1}(d)} \geq \Omega\left(\frac{\card{D}^n}{2^{c\log^{h-1}\card{\po t}}}\right).
\)
Thus, randomly choosing sufficiently many tuples from $D^n$,
say $\Omega\left({2^{c\log^{h-1}\card{\po t}}}\right)$ many of them,
with probability at least $1/2$ we will find a solution to $\po t(\o x) =d$,
if there is at least one.
This algorithm works then in time $O\left({2^{c\log^{h-1}\card{\po t}}}\right)$ which matches the complexity $\Omega\left({2^{c\log^{h-1}\card{\po t}}}\right)$ of the lower bound provided in Subsection \ref{sect-lower}, but possibly with a different constant $c$.

We start with observing that replacing $\po t(\o x)$ by the polynomial $\po t(\o x)-d$
we may assume that $d=0$.
Now starting with a single solution for the equation $\po t(\o x)=0$
we inductively create the sets $T^1,\ldots, T^h$ of solutions such that
$\card{T^k}\geq
\Omega\left(\frac{p_1^n\cdot \ldots\cdot p_k^n}{2^{c_k\log^{k-1}\card{\po t}}}\right)$.
It should be obvious that our final set $\card{T^h} \ci \po t^{-1}(0)$ witnesses that the size of $\po t^{-1}(0)$ is big enough.

We parameterise the sets $E^k$ defined in section \ref{sect-duper} by tuples $\o u \in D^n$
simply putting
\[
E^k(\o u)
= \set{\o b \in D^n : e_j(\o b) =e_j(\o u) \mbox{ \ for all \ } j\neq k}.
\]
Then, inside $E^k(\o u)$ we distinguish the subset
\[
E_0^k(\o a)= \set{\o b \in E^k(\o u) : \po t(\o b)=0}
\]
of solutions to our equation.
Then we fix one solution tuple $\o a \in \po t^{-1}(0)$
from which we will produce many other ones.
To do that we put $T^1= E_0^1(\o a)$
and $T^k=\bigcup_{\o u \in T^{k-1}} E_0^k(\o u)$.
It should be clear that any tuple in all $T^k$'s is a solution to our equation.
Thus, after showing that
$\card{E_0^k(\o u)} \geq \Omega\left(\frac{p_k^n}{2^{c_k\log^{k-1}\card{\po t}}}\right)$
we get that the size of $T^k$ is as big as promised,
so that we can conclude our proof.

Despite of our relativization of the $E^k$'s to the $E^k(\o u)$'s
(but keeping $\o u $ in the solution set $\po t^{-1}(0)$)
we still know that as long as $\o b \in E^k(\o u)$
the fact that $\po t(\o b)=0$ can be replaced (as previously)
by the system of only $k$ equations
$e_1\po t(\o b)=0,\ldots,e_k\po t(\o b)=0$,
where the normal forms $\po f_j$ for $e_j\po t$ reduce accordingly as in \eqref{rj}.
Thus to see that $\card{E_0^1(\o a)} \geq p_1^{n-1}$ note only that $E_0^1(\o a)$ consists of solutions to the linear equation
$0= c^1 +\sum_{i=1}^n  \lambda^1_i e_1(b_i)$.

Establishing the lower bound for $E_0^k(\o u)$ is more laborious.
We fix $\o u$ in $T^{k-1}$ (or more generally in $\po t^{-1}(0)$)
we repeat the procedure of section \ref{sect-duper} to produce
a relatively short (i.e. of size $O(\card{\po t}^{p_1})$)
polynomial
$\po t^\star_{\o u}(\o x)$ of $\dd$ that maps everything to only two values
$0,e_11$ and that depends only on $e_k(\o x)$,
and -- what is the most important -- has the property that over the set $E^k(\o u)$
the equations $\po t^\star_{\o u}(\o b)= e_11$ and $\po t(\o b)=0$ have exactly the same solutions.

As previously (in section \ref{sect-duper}) our goal is to rearrange polynomial $\po t^\star_{\o u}$
to a polynomial $\po t^\dstar_{\o u}$ that behaves on the set $\set{0, e_k1}$ like a conjunction and then apply SESH to the size of $\po t^\dstar_{\o u}$ to bound its arity.
On the way from $\po t^\star_{\o u}$ to $\po t^\dstar_{\o u}$ we create a polynomial
$\po t^\pstar_{\o u}$. To do that we refer to Lemma \ref{hyper}
(which is shown at the end of this section)
with $q=p_k$ and $Z=(\po t^\star_{\o u})^{-1}(e_11)\cap e_kD^n$
to get a hyperplane $H\ci e_kD^n$ of codimension $d\leq\log_{p_k}Z+p_k\log p_k$.
By Gauss elimination the set $\set{1,\ldots,n}$ can be split into two disjoint subsets $I,J$
with $\card{J}=d$ such that the hyperplane $H$ can be described by $d$ equations of the form
$x_j=\sum_{i\in I} \alpha^j_i x_i+\beta^j$, with the $\alpha$'s taken from $\m Z_{p_k}$,
while the $\beta$'s originally living in $\m Z_{p_k}$ are modified so that they are put into $e_kD$.
Now $\po t^\pstar_{\o u}$ is obtained from $\po t^\star_{\o u}$ by replacing $x_j$ with
$\sum_{i\in I} \alpha^j_i x_i+\beta^j$.
This slightly reduces the arity of $\po t^\pstar_{\o u}$ to be at least $n-\log_{p_k}Z-p_k\log p_k$
but $\card{\po t^\pstar_{\o u}}\leq O(n\cdot \card{\po t})\leq O(\card{\po t}^2)$.
However now the equation $\po t^\pstar_{\o u}(\o x)=e_11$ has exactly one solution
$\o b=(b_1,\ldots,b_{n-d})$, namely the one corresponding to the unique point in the intersection $Z\cap e_kD^n$.
To  make sure that $\po t^\dstar_{\o u}(x_1,\ldots,x_{x-d})$ behaves like a conjunction
we put
$\po t^\dstar_{\o u}(x_1,\ldots,x_{n-d})=
\po t^\dstar_{\o u}(x_1-e_k1+b_1,\ldots,x_{n-d}-e_k1+b_{n-d})$
and then turn $\po t^\dstar_{\o u}$ into a Boolean circuit of $(n-d)$-ary conjunction of size
$O(\card{\po t}^c)$ for some constant $c$.
This, by SESH gives that $\Omega\left(2^{c'\cdot (n-d)^{1/(k-1)}}\right)\leq O(\card{\po t}^c)$,
or in other words $n-c\log^{k-1}\card{\po t} \leq d$.
To conclude with our lower bound for $E^k(\o u)$ first note that this set fully corresponds to
$Z=(\po t^\star_{\o u})^{-1}(e_11)\cap e_kD^n$ so that  $\card{E^k_0(\o u)}=\card{Z}$.
Summing up we get
\[
n-c\log^{k-1}\card{\po t} \leq d \leq \log_{p_k}\card{Z}+c'=\log_{p_k}\card{E^k_0(\o u)}+c',
\]
and consequently
$\card{E^k_0(\o u)} \geq \Omega\left(\frac{p_k^n}{2^{c\log^{k-1}\card{\po t}}} \right)$,
as required.

\begin{lm}
\label{hyper}
For a non-empty subset $Z$ of the $n$-dimensional vector space $GF(q)^n$
there is an affine subspace $H$ of codimension at most $\log_q\card{Z}+q\log_2 q$
such that $\card{Z\cap H}=1$.
\end{lm}
\begin{proof}
We will successfully replace $Z$ by $Z\cap H$ where $H$ at the start is $GF(q)^n$.
A s long as $\card{Z} > q^{q-1}$ the set $Z$ has to contain at least $q$ linearly independent vectors, say $w_1,\ldots,w_q$. Now for a $q\times n$-matrix $W$ with rows $w_1,\ldots,w_q$
and the vector $a=(a_1,\ldots,a_q)\in GF(q)^q$ listing all elements of the field
the system of equations $W\cdot x = a$
has solutions, so that we pick one, say $[\alpha_1,\ldots,\alpha_n]$.
Consider $q$ hyperplanes determined by the equations of the form $\sum_i \alpha_ix_i=a_j$.
Note that each such hyperplane intersects $Z$, as $w_j$ belongs to such intersection.
Pick the one that leads to the intersection of the smallest size,
and replace $H$ by its intersection with this particular hyperplane.
Note that $Z\cap H$ has now at most $\frac{\card{Z}}{q}$ elements.

At some point we will arrive with $Z$ being too small to repeat this procedure.
So, if $\card{Z} \leq q^{q-1}$ but still $\card{Z} \geq 2$ we pick a coordinate
$i_0$ such that $Z$ contains at least two vectors that differ at this coordinate.
This time we consider all $q$ hyperplanes given by the equations $x_{i_0}=a_j$
and pick one that non-empty intersects $Z$ but this intersection is the smallest possible.
Replace $H$ with its intersection with this hyperplane.
Since that are at least two hiperplanes non-empty intersecting $Z$
we know that this time $\card{Z\cap H}\leq\frac{\card{Z}}{2}$.
\end{proof}

\section{The group case}
\label{sect-s4}

Both \csat{} and \ceqv{} are fully solved for groups.
The problems are polomial time solvable for nilpotent groups and \np/\conpc otherwise.
This is because a nilpotent groups are already supernilpotent.
However, as we have already mentioned, equations solving (not compressed by circuits)
may be still poly-time solvable -- this in fact is the case of the non-nilpotent group $\m S_3$. Actually, there are much more such examples \cite{foldvari-horvath}.
The smallest group for which the complexity is not known is the group $\m S_4$.
The method used in Section \ref{sect-example} can be almost directly applied
to provide an $\Omega(m^{c\log{m}})$ lower bound
for time complexity of solving equations (\polsat{})
and polynomials equivalence (\poleq{}), where $m$ is the size of the equation on input.

\begin{fact}\label{thm-s4}
The complexity of both \polsat{\m S_4} and \poleq{\m S_4} is $\Omega(m^{c\log{m}})$, where $m$ is the size of input (unless ETH fails).
\end{fact}
\begin{proof}
Before we start with the proof we note that $\set{1}<\m V<\m A_4<\m S_4$ is the full sequence of normal subgroup of $\m S_4$, where $\m V\simeq \m Z_2\times \m Z_2$ is the Klein group
and $\m A_4$ is the alternating group.
They correspond to the  levels $e_1(\m D)$, $e_2(\m D)$ and $e_3(\m D)$ of the algebra $\dd$
from Section \ref{sect-example}.

Below we summarize a few simple observations about the structure of $\m S_4$
and its normal subgroups:
\begin{itemize}
 \item $\m S_4/\m V\simeq \m S_3$,
 \item $[\m S_4,\m S_4]=\m A_4$,
 \item $[\m A_4,\m A_4]=\m V$
 \item $[\m V,\m V]=1$.
 \item for every $a\in \m A_4\setminus \m V$ we have that $[\m V,a]=\m V$,
\end{itemize}

We will show lower bound for \polsat{\m S_4}.
The proof for \poleq{} is nearly the same.
Let $c\in\m V\setminus\set{1}$.
Analogously as in our construction in Section \ref{sect-lower}
we start with a 3-CNF formula $\Phi$ (with $m$ clauses)
we constructs $\po t_{\Phi}$
such that $\Phi$ is satisfiable
iff $\po t_{\Phi}(y_1,y_2,y_3,y_4,x_1,\ldots,x_n)=c$ has a solution.

The construction of $\po t_{\Phi}$ is split into two steps.
To imitate $\pand^s_1(x_1,\ldots,x_s)$ we will use the $(s+4)$-ary terms
\begin{equation*}
\begin{split}
\alpha_s(y_1,y_2,y_3,y_4,&x_1,\ldots,x_s)=\\
&\alpha^{\circ}([[y_1,y_2],[y_3,y_4]],x_1,\ldots,x_s),
\end{split}
\end{equation*}
 where $\alpha^{\circ}(y,x_1,\ldots,x_s)=[[\ldots[y,x_1],\ldots],x_s]$.
%So our final equation will have a form
%\[\alpha_k(\o y,\o x)=c.\]
Note that, independently of how $y_1, y_2, y_3, y_4\in\m S_4$ are chosen
the value $[[y_1,y_2], [y_3,y_4]]$ is in $\m V$.
Moreover, any $d\in\m V$ can be realized as $[[y_1,y_2], [y_3,y_4]]$
for some $y_1,y_2,y_3,y_4\in\m S_4$.

Now we divide the clauses of $\Phi$ into $s$ parts,
each of which consist of at most $s$ clauses, say $\Phi_\ell$'s,
where $s\leq\lceil\sqrt{m}\rceil$.
We will imitate $\pcnf(\Phi_\ell)$ (with $n_\ell$ variables)
to code 3-CNF formula $\Phi_\ell$.
To do that we borrow (e.g. from \cite{gorazd-krzacz:preprimal} or from \cite{ik:lics18})
the polynomial $\po p_{\Phi_\ell}(x_1,\ldots,x_{n_\ell})$
(of exponential size in $n_\ell$) with range contained in $\m A_4$
whose behavior on each tuple $(x_1,\ldots,x_{n_\ell})\in \m S_4^{n_\ell}$ is,
modulo $\m V$, fully determined by the behavior of the $x_i$'s modulo $\m A_4$.
Namely $p_{\Phi_\ell}(x_1,\ldots,x_{n_\ell})\in\m V$ iff $\Phi(\bb(x_1),\ldots, \bb(x_n))=\perp$, where $\bb:S_4\mapsto \set{\perp, \top}$
is given by $\bb(x)=\top$ if $x\in\m A_4$ and $\bb(x)=\perp$ otherwise.
Now we put
$\po t_{\Phi}(\o y,\o x)$ to be
$\alpha_k(y_1,y_2,y_3,y_4,p_{\Phi_1}(\o x),\ldots,p_{\Phi_s}(\o x))$.

Suppose $t_{\phi}(\o y,\o x)=c$ for some $y$'s and $x$'s.
Indeed the fact that $[[y_1,y_2],[y_3,y_4]]\in\m V$
ensure us that none of the $p_{\Phi_\ell}(\o x)$'s might be $\m V$.
Consequently for all the  $\ell$'s we have
$\Phi_\ell(\bb(\o x))=\top$
so that $\Phi$ itself is satisfied while evaluated by $b(\o x)$.

Conversely, we translate a Boolean evaluation of the variables in $\Phi$ by the $z_i$'s,
to a corresponding evaluation of the $x_i$'s by elements of $\m S_4$
so that we chose $x_i\in\m A_4$ whenever $z_i=\top$,
and all the other $x_i$'s are chosen from outside $\m A_4$.
Obviously all the $\po p_{\Phi_\ell}$'s
are then put inside $\m A_4$ but outside $\m V$.
We are left with finding values for the $y_i$'s.
But, using the fact that for $[\m V,a]=\m V$ for any $a\in\m A_4-\m V$
and knowing that the $\po p_{\Phi_\ell}(\o x)$'s are in this difference,
we find  $u\in\m V$ so that
\[
\alpha^{\circ}(u,p_{\Phi_1}(\o x),\ldots,p_{\Phi_s}(\o x))=c.
\]
Now, this $u$ can be decomposed into $u=[[y_1,y_2],[y_3,y_4]]$ for some
$y_1,\ldots,y_4\in S_4$.

Finally we refer to ETH and argue like at the beginning of Section \ref{sect-lower}
to get the promised lower bound for equation solution in $\m S_4$.
\end{proof}

Solvable but not non-nilpotent gap in equation solving for groups
is open for about 20 years since Goldmann's and Russel's paper \cite{goldman-russell}.
Since then, a lot of effort has been put into finding new classes of solvable but non-nilpotent groups for which \polsat{} and \poleq{} are in \ptime
(e.g. \cite{horvath-szabo:groups}, \cite{horvath:metabelian}, \cite{foldvari-horvath}).
The group $\m S_4$ is now the first known example of a solvable but non-nilpotent group
for which probably do not exist polynomial time algorithms solving these problems.
Moreover, our method used in the proof of Fact \ref{thm-s4}
is quite general and can be used for showing lower bounds for other groups
or even for other solvable but non-nilpotent algebras from congruence modular varieties.

In fact, very recently Armin Wei{\ss} presented a proof
\cite{weiss-eth}
that under ETH neither \polsat{} nor \poleq{} can be in \ptime
for the solvable groups that are not 3-step supernilpotent
(or even not 2-step supernilpotent, but with an additional technical assumption).
Note here that the concept of the $h$-step supernilpotency in groups coincides with the
one of the Fitting length $h$.
Combining his and our efforts now we can remove this artificial technical assumption
and actually strengthen the lower bound to be read:
\\
{\em If $\m G$ is a finite solvable nonnilpotent group of Fitting length $h>2$
then both \polsat{G} and \poleq{G} require at least $O(2^{c \log^{h-1} m})$ steps,
where $m$ is the length of the polynomial(s) on input (unless ETH fails).}

\section{Conclusions}
\label{sect-concl}

We propose a couple of methods that are highly effective
in filling the nilpotent versus supernilpotent gap
for the problems  $\csat{}$ and $\ceqv{}$,
but with the help of two strong complexity hypothesis.
Our methods are particularly effective for $h$-step supernilpotent algebras for $h\geq 3$.
However these methods do not fully solve the problems for $2$-step supernilpotent algebras
(as they lead only a probabilistic upper bound, and this bound relies on SESH).

Since supernilpotent algebras do already
have polynomial time algorithms for both $\csat{}$ and $\ceqv{}$,
it seems that the $2$-step supernilpotent ones
form the natural next step to be attacked.
All the known to us examples of such algebras,
including the $\dpq$'s, lie on the polynomial side
(without any additional complexity hypothesis).
Moreover \cite{kkk} contains a proof that $\ceqv{}$ for $2$-nilpotent algebras is in \ptime.
As $2$-nilpotent algebras are $2$-step supernilpotent
this still leaves the hope that
the last ones also lie on the polynomial side.
Also, Theorem \ref{thm-puper} provides a polynomial randomized upper bound for $h=2$.
This makes our hope even stronger.

On the other hand we do hope that the boundary between tractable and hard algebras
is determined by this new measure of failure of the supernilpotency,
as there are examples \cite{ikk:strat}
of $3$-nilpotent but not $2$-nilpotent algebras with polynomially solvable
$\csat{}$ and $\ceqv{}$. They are $2$-step supernilpotent.

\medskip

In our second remark we note that both our algorithms can be parameterized
by the (lower) bound for the conjunction-like polynomials or $CC^0$-circuits.
If the lower bound provided by SESH is replaced by a computable but slower growing function $f(n)$
Then our method gives
\begin{itemize}
  \item a deterministic algorithm of complexity $O(n^{c\cdot f^{-1}(\card{\po t}^d)})$,
  \item a randomized algorithm of complexity $O(2^{c\cdot f^{-1}(\card{\po t}^d)})$,
\end{itemize}
where $\card{\po t}$ is the size of polynomial or a circuit,
$n$ is the number of variables (or input gates),
and $c,d$ are some constants.
This shows a very strong connections between the complexity of $\csat{}$ and $\ceqv{}$
and the size in which conjunctions can be expressed by $CC^0$-circuits (or polynomials).

In particular if $f(n)\geq 2^{cn}$ for some $c>0$
(which is true in $\dpq$, i.e. in the case of $h=2$)
then from what we said above the proof of Theorem \ref{thm-puper}
supplies us with a polynomial time randomized algorithms.
In the case $h=1$, i.e. for supernilpotent algebras, there is even no such function $f$,
as there is a bound for the arity of polynomials that expres conjunction-like behavior.
In this case we can slightly modify the method used in the proof of Theorem \ref{thm-puper}
to get linear algorithms for $\csat{}$ and $\ceqv{}$.
In particular (as nilpotent groups are supernilpotent) we get
a striking division between untractable (\np/\conpc) non-nilpotent groups
and the nilpotent ones that can be treated in probabilistic linear time \cite{kk-linear}.

The other feature provided by our proof of Theorem \ref{thm-puper}
tells us that a short polynomial splits its domain into rather large subsets
on which it is constant.
In particular it is not possible to separate, by polynomials,
not only single points (what is usually done by a conjunction-like function)
but even larger subsets in the big powers of the algebra.

\end{document}